\documentclass[preprint,a4paper,fleqn]{cas-sc}

\ExplSyntaxOn \cs_gset:Npn \__first_footerline: { \group_begin: \small \sffamily \__short_authors: \group_end: } \ExplSyntaxOff

\usepackage[numbers]{natbib}

\def\tsc#1{\csdef{#1}{\textsc{\lowercase{#1}}\xspace}}
\tsc{WGM}
\tsc{QE}

\newtheorem{theorem}{Theorem}
\newtheorem{corollary}{Corollary}
\newtheorem{proposition}{Proposition}
\newtheorem{lemma}[theorem]{Lemma}
\newdefinition{remark}{Remark}
\newdefinition{definition}{Definition}
\newproof{proof}{Proof}
\newtheorem{example}{Example}

\usepackage{graphicx}
\usepackage{amsmath}
\usepackage{amssymb}
\usepackage{booktabs}
\usepackage{multicol}
\raggedcolumns
\usepackage{enumitem}
\usepackage{xcolor}
\usepackage{tikz}
\usetikzlibrary{arrows.meta, positioning}

\begin{document}
\let\WriteBookmarks\relax
\def\floatpagepagefraction{1}
\def\textpagefraction{.001}

\shorttitle{Implication Problems over Positive Semirings}    

\shortauthors{M. Hirvonen}  

\title [mode = title]{Implication Problems over Positive Semirings}  

\tnotemark[1] 

\tnotetext[1]{This is a completely revised and expanded version of a paper presented at the 14th International Symposium on Foundations of Information and Knowledge Systems (FoIKS 2026), which was held at the Leibnizhaus, Hanover, March 23–26, 2026.} 

%

\author[1]{Minna Hirvonen}[orcid=0000-0002-2701-9620]

\cormark[1]


\ead{minna.hirvonen@thi.uni-hannover.de}

\ead[url]{}


\affiliation[1]{organization={Leibniz Universität Hannover},
            addressline={Appelstraße 9a}, 
            city={Hannover},
            postcode={30167}, 
            country={Germany}}







\cortext[1]{Corresponding author}



\begin{abstract}
We study various notions of dependency in semiring team semantics. Semiring teams are essentially database relations, where each tuple is annotated with some element from a positive semiring. We consider semiring generalizations of several dependency notions from database theory and probability theory, including functional and inclusion dependencies, marginal identity, and (probabilistic) independence.
We examine axiomatizations of implication problems, which are rule-based characterizations for the logical implication and inference of new dependencies from a given set of dependencies. Semiring team semantics provides a general framework, where different implication problems can be studied simultaneously for various semirings. The choice of the semiring leads to a specific semantic interpretation of the dependencies, and hence different semirings offer a way to study different semantics  (e.g., relational, bag, and probabilistic semantics) in a unified framework. 
\end{abstract}




\maketitle





\def\dep{=\!\!}
\newcommand{\Var}{\mathrm{Var}}
\newcommand{\Supp}{\mathrm{Supp}}
\newcommand{\Cl}{\mathrm{Cl}}
\newcommand{\Const}{\mathrm{Const}}
\newcommand{\scc}{\mathrm{scc}}
\newcommand{\pia}{\perp}
\newcommand{\X}{\mathbb{X}}

\section{Introduction}

Various notions of dependency appear in many research areas, such as mathematics, statistics, and computer science. 
For example, the notions of functional and
inclusion dependencies are applied in relational databases as primary and foreign key constraints to
maintain data integrity, and the notions of independence and marginal identity of random variables (IID variables) are used in probability theory and statistics to simplify the underlying mathematics.
Many of these dependency notions share analogous properties across different domains of application. For example, the complete axiomatizations for inclusion dependencies over relational databases \cite{casanova82} and for marginal identities over probability distributions \cite{HANNULA22} are very similar, and the semigraphoid axioms \cite{pearl88} are sound for both embedded multivalued dependency and conditional independence.

It would therefore be useful to study these notions in a general framework that could help explain how the characteristics of different application domains affect the interpretation and properties of the notions.
One possible way to do this is to consider \emph{semirings}, which is a generalization of rings without the requirement that every element must have an additive inverse. 
The $K$-relations, introduced in \cite{green07}, are relations such that each of their tuples is annotated by some element from a semiring $K$. In \cite{green07}, $K$-relations over commutative semirings were used for \emph{data provenance}, which refers to meta-information obtained about the origin and processing history of data based on propagation of tuple annotations through queries, but $K$-relations have also been studied in relation to dependencies. For example, conditional independence, functional dependency, marginal identity, and multivalued dependency over semirings have been studied in \cite{hannula:LIPIcs.ICDT.2024.20,hannula2023conditionalindependencesemiringrelations}, and join operations over $K$- relations in \cite{atserias20}. 

In \cite{barlag23}, semirings were introduced into the framework of \emph{team semantics}, which is a semantical system for logics that can express different dependencies using new kinds of atomic formulas. This is possible because in team semantics, the satisfaction of logical formulas is determined on a set of assignments, instead of just a single assignment, as, e.g., in the case of the usual semantics of first-order logic. These sets of assignments are called teams, and they are essentially database relations. The most well-known logic based on team semantics is the \emph{dependence logic} \cite{vaananen07} that extends first-order logic with dependency atoms that correspond to the notion of functional dependency. Other well-known variants include the inclusion logic \cite{galliani12} and the independence logic \cite{Gradel2013-GRDDAI}. The framework of team semantics has been generalized to multi-team (bag) semantics \cite{DurandHKMV18} and probabilistic team semantics \cite{HKMV18}, and also extended to various logics, including propositional logic \cite{YangV16}, modal logic \cite{vaananen2008modal}, and inquisitive logic \cite{ciardelli11}.

Recently, the notion of $K$-teams was introduced as a generalization that covers the usual relational team semantics, multi-team semantics, and probabilistic team semantics in a unified framework \cite{barlag23}. The goal of this paper is to define and study axiomatizations over general $K$-teams for various classes of atoms that have previously been studied separately in these different variants of team semantics. We use the terminology of $K$-teams, as defined in \cite{barlag23}, but since $K$-teams are essentially annotated relations, the results also apply to $K$-relations. In particular, we study the logical inference of atoms in the form of \emph{implication problems}. The implication problem for a class of atoms $\mathcal{C}$ over a semiring $K$ is the problem of deciding whether, for a given finite set $\Sigma\cup\{\tau\}$ of $\mathcal{C}$ atoms, any $K$-team that satisfies every atom of $\Sigma$ also satisfies $\tau$.

One of the main results of this paper is that the axiomatization of the implication problem for functional dependencies and the unary variants of marginal identity and marginal distribution equivalence over probabilistic teams introduced in \cite{hirvonen24} is sound and complete more generally for any additively cancellative positive semiring. We also consider unary  marginal identity and distribution equivalence together with independence atoms, and show that for these dependencies, the axiomatization from \cite{hirvonen24conf} over probabilistic teams similarly extends to any additively and multiplicatively cancellative positive commutative semiring.

We also analyze a concept of weighted inclusion dependency that on suitable semirings coincides with the notions of inclusion dependency and marginal identity. The weighted inclusion dependency has a certain symmetry property exactly on additively cancellative positive semirings, and the well-known axioms for the usual inclusion dependencies form a sound and complete axiomatization for the implication problem of weighted inclusion dependencies over a certain class of non-additively cancellative semirings \cite{hannula2026dichotomyaxiomatisinginclusiondependencies, hirvonen26}\footnote{Hannula et al. \cite{hannula2026dichotomyaxiomatisinginclusiondependencies} axiomatized the implication problem for (weighted) inclusion dependencies in the setting of $K$-databases with multiple relations, where $K$ has the 
natural order $\leq_{\rm n}$ defined such that $\forall k,k'\in K$, $k\leq_{\rm n} k'$ iff $\exists h\in K$ s.t. $k+h=k'$. 
Since we are working with $K$-teams that are essentially unirelational databases, the result of \cite{hannula2026dichotomyaxiomatisinginclusiondependencies} generalizes Theorem \ref{INDax} that was originally published in the conference version \cite{hirvonen26} of this paper. Note that in this paper, we work with so-called zero-min orders; when the semiring $K$ is positive, the natural order is a zero-min order, but there are other zero-min orders besides the natural order. (See Appendix \ref{exapp}, Propositions \ref{polynom} and \ref{doubletwo} for some examples.) This variation with the orders of $K$ is partly responsible for some of the differences in the results between the papers, in particular with respect to the symmetry characterizations. For more details, see the footnote for Proposition \ref{addcan_prop}.}.
We consider the unary variant of weighted inclusion dependency together with (unary) functional dependencies and certain decidable subclasses of conditional independence atoms
and obtain sound and complete axiomatizations over suitable classes of non-additively cancellative semirings. 
We also introduce a notion of marginal distribution inclusion, which is a non-symmetric variant of marginal distribution equivalence atom, and present a sound and complete axiomatization for its unary variant over a certain class of non-additively cancellative semirings. Some of the axiomatizations in this paper contain an infinite cycle rule schema, but we show that the infiniteness is unavoidable as the respective implication problems are not finitely axiomatizable.  Tables \ref{tab:results} and \ref{tab:semiringprops} summarize the results for different implication problems.

\begin{table}
\begin{center}
\begin{tabular}{ c|c|c|c|c } 
Implication problem & \ Complete ax. & \ Finite ax. \ & \ Armstrong rel. \ & \ Complexity \ \\
 \hline
 IND & yes \cite{hannula2026dichotomyaxiomatisinginclusiondependencies} ($\star$) &  yes \cite{hannula2026dichotomyaxiomatisinginclusiondependencies} ($\star$) &  yes ($\star$) & \textsf{PSPACE}-compl. ($\mathbb{B}$) \cite{casanova82}\\
 MDI & ? &  ? & no ($\mathbb{B}$) ($\star$) & ? \\
 MI  & yes \cite{hannula2023conditionalindependencesemiringrelations} & yes \cite{hannula2023conditionalindependencesemiringrelations} &  ? & ? \\
 MDE  & ? & ? & ? & ? \\
 UMDI & yes ($\star$) &  yes ($\star$) & no ($\mathbb{B}$) ($\star$) & in linear time ($\star$) \\
 FD+IND & no ($\mathbb{B}$) \cite{chandra85} &  no ($\mathbb{B}$) \cite{chandra85} & ? & undec. ($\mathbb{B}$) \cite{chandra85} \\
 FD+IA & \ no ($\mathbb{R}_{\geq0}$) \cite{LI23}\ &  \ no ($\mathbb{R}_{\geq0}$) \cite{LI23}\ & ? & undec. ($\mathbb{R}_{\geq0}$)  \cite{LI23}\\
 FD+UIND& yes ($\star$) &  no ($\star$) & ? & in \textsf{P} ($\star$) \\
  FD+SCI  & yes \cite{hannula2023conditionalindependencesemiringrelations} & yes \cite{hannula2023conditionalindependencesemiringrelations} & ?  & in \textsf{P} \cite{hannula2023conditionalindependencesemiringrelations}  \\ 
  FD+UIND+SCI & yes ($\star$) &  no ($\star$) & ? & in \textsf{P} ($\star$) \\
    FD+UMI+SCI & yes ($\star$) &  no ($\star$) & ? & in \textsf{P} ($\star$) \\
 FD+UMI+UMDE & yes ($\star$) &  no ($\star$) & yes ($\star$) & in \textsf{P} ($\star$) \\
 UFD+UIND+IA & yes ($\star$) &  no ($\star$) &  yes ($\star$) & in \textsf{P} ($\star$) \\
  UFD+UMI+IA & yes ($\star$) &  no ($\star$) &  ? & in \textsf{P} ($\star$) \\
 UMI+UMDE+IA \ & yes ($\star$) &  yes ($\star$) & ? & in \textsf{P} ($\star$) \\
 FD+UMI+UMDE+SCI \ & ? &  ? & ? & ? \\
 UFD+UMI+UMDE+IA \ & ? &  ? & ? & ? \\
 \hline
\end{tabular}
\end{center}
    \caption{The table summarizes results for different implication problems: existence of a complete axiomatization/finite axiomatization/Armstrong relations, and the complexity of the implication problem. Each ``yes'' or inclusion in a complexity class refers to the result over the class of semirings described in Table \ref{tab:semiringprops}. In contrast, each ``no'' or hardness result is only known to hold for the semiring in parentheses, unless no semiring is specified, in which case the result holds over the class described in Table \ref{tab:semiringprops}.  The question mark symbol ``\ ?\ '' means that the problem is open. 
    The results marked with ($\star$) are from this paper.}
    \label{tab:results}
\end{table}

 \begin{table}  
\begin{center}
\begin{tabular}{ c|c|c|c|c|c } 
Implication problem & \ zero-min \ & $+$ & $\times$ & \ $k+k=k$ \ & \ $k+h=k$ \ \\
 \hline
  IND & any  & (no) & - & yes & (yes)  \\
IND \cite{hannula2026dichotomyaxiomatisinginclusiondependencies} & natural  & - & - & - & -  \\
 UMDI  & n/a & (no) & - & - & \  yes \ \\
 MI \cite{hannula2023conditionalindependencesemiringrelations} & n/a & - & - & - & - \\
FD+UIND & any & (no) & -  & yes & (yes) \\ 
 FD+SCI \cite{hannula2023conditionalindependencesemiringrelations} & n/a & - & yes  & - & - \\ 
FD+UIND+SCI & any & (no) & yes  & yes & (yes) \\ 
 FD+UMI+SCI & n/a & (no) & yes  & yes & (yes) \\ 
 UFD+UIND+IA & any & (no) & yes & \ yes \ & (yes) \\ 
 UFD+UMI+IA & n/a & (no) & yes & \ yes \ & (yes) \\ 
 FD+UMI+UMDE  & n/a & \ yes \ & - & (no) & (no) \\
 UMI+UMDE+IA \ & n/a & \ yes \ & \ yes \ & (no)  & (no) \\
 \hline
\end{tabular}
\end{center}
    \caption{The table describes for each implication problem the semiring class for which the existence of a sound and complete axiomatization has been shown. 
    The ``zero-min'' refers to the type of zero-min order of the semiring, 
    ``$+$'' and ``$\times$'' refer to the additive cancellativity, and commutativity and multiplicative cancellativity, ``$k+k=k$'' and ``$k+h=k$'' refer to self absorptivity and  weak absorptivity. 
    In the zero-min columnm, ``any'' refers to any zero-min order, ``natural'' refers to the natural order, and ``n/a'' means that no order is needed. In the other columns, ``yes''/``no'' in a column means that the property is required/not allowed, and the hyphen symbol ``\ - \ '' means that either option is possible. Column values implied by another column are in parentheses.
    }
    \label{tab:semiringprops}   
\end{table}

\section{Preliminaries}

This section contains the preliminaries needed for the semiring semantics of dependencies; we go through the basic definitions for semirings, $K$-teams, and the dependency notions considered in this paper. 

\subsection{Semirings}

We begin the section with the definition of a semiring, then consider some useful properties and give examples of semirings that appear in different contexts. Semirings are structures that correspond to rings without the requirement that every element must have an additive inverse. The formal definition is as follows.
\begin{definition}
  A \emph{semiring} $K$ is a tuple $(K,+,\times,0,1)$ such that
  \begin{enumerate}
      \item[(i)] $+$ and $\times$ are binary operations on $K$,
      \item[(ii)] $(K,+,0)$ is a commutative monoid with identity element 0,
      \item[(iii)] $(K,\times,1)$ is a monoid with identity element 1,
      \item[(iv)] $h\times(k+k')=h\times k + h\times k'$ and $(k+k')\times h=k\times h + k'\times h$ for all $k,k',h\in K$,
      \item[(v)] $0\times k=k\times 0=0$ for all $k\in K$.
  \end{enumerate}
  \end{definition}
A semiring $K$ is called \emph{commutative} if $(K,\times,1)$ is a commutative monoid. A semiring $K$ is \emph{positive} if $k+k'=0$ implies $k=k'=0$ and $k\times k'=0$ implies $k=0$ or $k'=0$. A semiring $K$ is \emph{additively cancellative} if $h+k=h+k'$ implies $k=k'$, and \emph{multiplicatively cancellative} if $h\times k=h\times k'$ implies $k=k'$, and $k\times h=k'\times h$  implies $k=k'$ for all $h\neq 0$. A semiring $K$ is \textit{weakly additively cancellative} if $k+h=k$ implies $h=0$. Note that if $K$ is additively cancellative, then it is also weakly additively cancellative.

We say that an element $k\in K$ is \emph{(additively) idempotent} if $k+k=k$. Note that if $K$ weakly additively cancellative, then 0 is the only idempotent element of $K$. If $K$ is multiplicatively cancellative and contains an idempotent nonzero element, then $K$ is idempotent, i.e., all its elements are idempotent. We say that $K$ is \textit{self absorptive} if it contains an additively idempotent nonzero element, and we say that $K$ is \textit{weakly absorptive} if there are nonzero $k,h\in K$ such that $k+h=k$. Note that if $K$ is positive, then it is weakly absorptive if and only if it is not weakly additively cancellative.

 A semiring is \emph{ordered} if there exists a reflexive, transitive, and antisymmetric relation $\leq$ (i.e., a partial order) such that for all $k,k', h \in K$ $k\leq k'$ implies $k+h \leq k'+h$ and ($k\leq k'$ and $0 \leq h$) imply ($k\times h \leq k'\times h$ and $h\times k \leq h\times k'$). We write $k<k'$ if $k\leq k'$ and $k\neq k'$. An order $\leq$ is \emph{total} if for all $k,k'\in K$ $k\leq k'$ or $k'\leq k$. We say that an order on $K$ is \emph{zero-min order} if $k\leq 0$ implies $k=0$ for all $k\in K$, i.e., 0 is a minimal element of the order. If the order is both total and zero-min, then 0 is the only minimal element.
 
 The \emph{natural preorder} $\leq_{\rm n}$ of $K$ is defined by $k\leq_{\rm n} k'$ iff there is $h\in K$ such that $k+h=k'$. This relation is always reflexive and transitive, and if $\leq_{\rm n}$ is also antisymmetric, the semiring $K$ is ordered by $\leq_{\rm n}$, which is then called the \emph{natural order} of $K$. If the semiring $K$ ordered by $\leq_{\rm n}$ is positive, then $\leq_{\rm n}$ is also a zero-min order on $K$. Note also that $K$ is always ordered by $\leq_{\rm n}$ if $K$ is positive and additively cancellative. 
 
\begin{example}
The following structures are examples of semirings\footnote{Note that the minimum and the maximum in (v)--(viii) are taken with respect to the usual total order of (the extended) reals. The (possibly partial) order of the semiring in general does not need to be the same as this, as long as it respects these operations defined via the usual order of reals, as required in the definition of ordered semiring.}.
\begin{itemize}
    \item[(i)] The \emph{Boolean semiring} $\mathbb{B}=(\mathbb{B},\lor,\land,0,1)$ is the two-element Boolean algebra that models logical truth.
    \item[(ii)] The \emph{semiring of natural numbers} $\mathbb{N}=(\mathbb{N},+,\times,0,1)$ consists of the natural numbers with standard addition and multiplication.
    \item[(iii)] The \emph{probability semiring} $\mathbb{R}_{\geq0}=(\mathbb{R}_{\geq0},+,\times,0,1)$ consists of the non-negative reals with standard addition and multiplication.
    \item[(iv)] The \emph{semiring of multivariate polynomials} $\mathbb{N}[X]=(\mathbb{N}[X],+,\times,0,1)$ that consists of multivariate polynomials with natural number coefficients and their standard addition and multiplication.
    \item[(v)] The \emph{tropical semiring} $\mathbb{T}=(\mathbb{R}\cup\{\infty\},\min,+,\infty,0)$ consists of the expanded reals with minimum and standard addition as addition and multiplication, respectively.
    \item[(vi)]The \emph{Viterbi semiring} $\mathbb{V}=([0,1],\max,\times,0,1)$ consists of the unit interval with maximum and standard multiplication as addition and multiplication, respectively.
    \item[(vii)] The \emph{fuzzy semiring} $\mathbb{F}=([0,1],\max,\min,0,1)$ consists of the unit interval with maximum and minimum as addition and multiplication, respectively. 
    \item[(viii)] The \emph{Lukasiewicz semiring} $\mathbb{L}=([0,1],\max,\times,0,1)$ consists of the unit interval with maximum and multiplication $k\times k' := \max\{0,k + k'-1\}$ as addition and multiplication, respectively.
\end{itemize}
\end{example}
The above semirings have applications in various contexts. The Boolean semiring models logical truth and can be used as a framework for classical decision problems \cite{KR2026-12}, the semiring of natural numbers relates to bag semantics and problems with counting \cite{KR2022-10}, and the probability semiring can model probabilities \cite{DERKINDEREN2024109130} or problems with geometric features \cite{KR2026-12,schaefer2024existentialtheoryrealscomplexity}. The semiring of multivariate polynomials is relevant in provenance analysis \cite{green07}. The tropical and Viterbi semirings are used in optimization, e.g., performance analysis \cite{DBLP:journals/access/OmanovicOC23}, reachability problems \cite{DBLP:journals/ijac/GaubertK06}, and in dynamic programming \cite{1054010}, and the fuzzy and
Lukasiewicz semirings are used in multivalued logic \cite{hajek1998metamathematics}.

  The Boolean semiring, the tropical semiring, and the Viterbi semiring are commutative, positive, and multiplicatively cancellative, but not weakly additively cancellative. The semiring of natural numbers, the probability semiring, and the semiring of multivariate polynomials with natural number coefficients are commutative, positive, and additively and multiplicatively cancellative. The fuzzy semiring is commutative and positive, but not  weakly additively cancellative or multiplicatively cancellative. The Lukasiewicz semiring is commutative, but not positive or weakly additively cancellative or multiplicatively cancellative. The Boolean semiring, the tropical semiring, the Viterbi semiring, the fuzzy semiring, and the Lukasiewicz semiring are idempotent. Note that the tropical semiring with the usual order of the extendend reals is not zero-min ordered, because its zero element is $\infty$.

In this paper, we only consider positive nontrivial semirings, i.e., $0\neq 1$, so our results do not cover the Lukasiewicz semiring, which is not positive. Since the positiveness of the semiring is assumed throughout this paper, we do not mention it for each theorem statement separately. Additional assumptions are necessary for some of the results, but these will be specified later in the relevant sections of the paper.

\subsection{$K$-Teams and Dependencies}

We now consider some notions that are needed for defining semiring teams and dependencies over them. After the definitions, we will also give some examples. Let $K$ be a positive (and commutative, totally zero-min ordered, multiplicatively cancellative\footnote{See Remark \ref{remark}.}) semiring, $D$ a finite set of variables, and $A$ a finite set of values. A \textit{tuple} of variables is a finite sequence $\bar{x}=x_1\dots x_n$. The notation $\bar{x}\setminus\bar{y}$ refers to the tuple obtained from $\bar{x}$ by removing the variables that appear in the tuple $\bar{y}$, e.g., if $\bar{x}=x_1\dots x_6$ and $\bar{y}=x_3x_4x_6$, then $\bar{x}\setminus\bar{y}=x_1x_2x_5$. We apply an analogous notation for tuples of values using letters from the beginning of the alphabet, e.g., $\bar{a}=a_1\dots a_n$ and $\bar{a}\setminus\bar{b}$. An \emph{assignment} s is a function $s\colon D\to A$. A \emph{team} $X$ of $A$ over $D$ is a finite set of assignments $s\colon D\to A$. We may assume that the set $D$ is totally ordered, e.g., $D=\{x_1,\dots,x_n\}$ and hence an assignment $s\in X$ can also be viewed as the tuple $s(x_1)\dots s(x_n)$. A $K$-team is a function $\mathbb{X}\colon X\to K$. The value $\mathbb{X}(s)$ is called the \emph{weight} of the assignment $s$.

For any subset $V\subseteq D$, the \emph{restriction} 
of $\mathbb{X}$ to $V$ is the team $\mathbb{X}\restriction V\colon X\restriction V\to K$, where $X\restriction V=\{s\restriction V:s\in X\}$ and $\mathbb{X}\restriction V(s)=\sum_{\substack{s'\in X,\\s'\restriction V=s}}\X(s)$. The notation $s\restriction V$ refers to the usual notion of the restriction of the function $s$ to $V$.
The \emph{support} of a $K$-team $\mathbb{X}$ is the set of those assignments that have nonzero weight, i.e.,  $\Supp(\mathbb{X}):=\{s\in X: \mathbb{X}(s)\neq 0\}$. Let $\bar{x}$ and $\bar{a}$ be a tuple of variables and a tuple of values, respectively. The \emph{marginal} of $\mathbb{X}$ over $\bar{x}=\bar{a}$ is defined as $|\mathbb{X}_{\bar{x}=\bar{a}}|:=\sum_{\substack{s\in X,\\s(\bar{x})=\bar{a}}}\mathbb{X}(s)$,
where the empty sum is interpreted as 0. We denote by $|\mathbb{X}|$ the weight of the entire $K$-team, i.e., $|\mathbb{X}|:=\sum_{s\in X}\mathbb{X}(s)$. Define also $X(\bar{x}):=\{\bar{a}\in A^{|\bar{x}|}:|\mathbb{X}_{\bar{x}=\bar{a}}|\neq 0\}$.

A multiset $(B,m)$ is a set equipped with a multiplicity function $m\colon B\to\{1,2,\dots\}$ such that $m(b)$ is the number of copies of the element $b\in B$ in the multiset. We write $b\in(B,m)$ if $b\in B$, and $(B',m')\subseteq(B,m)$ if $B'\subseteq B$ and $m'(b)\leq m(b)$ for all $b\in B'$. For simplicity, we often write multisets with double wave brackets, e.g., $(\{b,b'\},m)$ such that $m(b)=1$ and $m(b')=2$ can be represented as $\{\{b,b',b'\}\}$.

A tuple $\bar{x}=x_1\dots x_n$ is said to be \emph{repetition-free} if $x_i= x_j$ implies $i= j$ for all $i,j\in\{1,\dots,n\}$.  Let $\Var(\bar{x})$ denote the set of variables that appear in the tuple $\bar{x}$. We say that the tuples $\bar{x}$ and $\bar{y}$ are \emph{disjoint} if $\Var(\bar{x})\cap\Var(\bar{y})=\emptyset$. The length of the tuple $\bar{x}$ is denoted by $|\bar{x}|$.
Let $\bar{x}$ and $\bar{y}$ be tuples of variables. Then $\dep(\bar{x},\bar{y})$ is the (functional) dependency (FD) atom. If the tuple $\bar{x}$ is empty, we write $\dep(\bar{y})$ instead of $\dep(\emptyset,\bar{y})$ and say that $\dep(\bar{y})$ is a constancy atom (CA). If the tuples $\bar{x}$ and $\bar{y}$ are repetition-free, then $\bar{x}\subseteq^*\bar{y}$ and $\bar{x}\approx^* \bar{y}$ are the marginal distribution inclusion (MDI)\footnote{In the previous conference version of this paper marginal distribution inclusion atoms were abbreviated by IND$^*$ instead of MDI. The idea behind the previous abbreviation was to use $*$ to refer to the star symbol in the notation $\bar{x}\subseteq^*\bar{y}$, and to the fact that both $\bar{x}\subseteq^*\bar{y}$ and $\bar{x}\approx^*\bar{y}$ are ``distribution'' variants of the inclusion atom $\bar{x}\subseteq\bar{y}$ and the marginal identity $\bar{x}\approx\bar{y}$, respectively. Since there are other ways to define similar atoms (see Remark \ref{othervariant}), we no longer use the abbreviation IND$^*$ as it is not clear why this version should be regarded as \emph{the} ``star variant'' of IND.}, and the marginal distribution equivalence (MDE) atoms, respectively. 
If the tuples $\bar{x}$ and $\bar{y}$ are repetition-free and of the same length, then $\bar{x}\leq \bar{y}$ and $\bar{x}\approx \bar{y}$ are the (weighted) inclusion dependency (IND) and the marginal identity (MI) atom, respectively. 
If $\bar{x}$, $\bar{y}$, and $\bar{z}$ are tuples\footnote{It is common to only allow disjoint tuples in the definition of IAs and CIs. In the usual team semantics setting, non-disjoint tuples are usually allowed, so we allow then here as well. Note that the semantics of the CIAs are defined here such that $\bar{y}\perp_{\bar{x}}\bar{y}$ is equivalent to $\dep(\bar{x},\bar{y})$.} (not necessarily of the same length), then $\bar{y}\perp_{\bar{x}}\bar{z}$ is a conditional independence (CI) atom. If $\bar{x}$ is empty, then instead of $\bar{y}\perp_{\emptyset}\bar{z}$, we write $\bar{y}\perp \bar{z}$ and call this the independence atom (IA). If $\Var(\bar{x}\bar{y}\bar{z})=D$, then $\bar{y}\perp_{\bar{x}}\bar{z}$ is a \emph{saturated} conditional independence (SCI) for any $K$-team $\X$ over $D$.
If $|\bar{x}|=|\bar{y}|=|\bar{z}|=1$, then the corresponding atom is called \emph{unary}, and the letter ``U'' is added in front of the abbreviation of the atom name, e.g., UFD, UIND, UMDE etc. A constancy atom $\dep(\bar{y})$ is unary if $|\bar{y}|=1$, and UCAs are considered UFDs in this work.

We extend the notation $\Var(\bar{x})$ to atoms $\sigma$ and sets of atoms $\Sigma$ in the obvious way. Let $\sigma$ be an atom and $\mathbb{X}$ a $K$-team of $X$ over $D$ such that $\Var(\sigma)\subseteq D$. The satisfaction relation $\mathbb{X}\models\sigma$ is defined as follows.
\begin{itemize}
    \item $\mathbb{X}\models\dep(\bar{x},\bar{y})$ iff for all $s,s'\in\Supp(\mathbb{X})$, $s(\bar{x})=s'(\bar{x})$ implies $s(\bar{y})=s'(\bar{y})$,
    \item $\mathbb{X}\models \bar{x}\leq \bar{y}$ iff $|\mathbb{X}_{\bar{x}=\bar{a}}|\leq|\mathbb{X}_{\bar{y}=\bar{a}}|$ for all $\bar{a}\in A^{|\bar{x}|}$,
   \item $\mathbb{X}\models \bar{x}\subseteq^* \bar{y}$ iff $\{\{|\mathbb{X}_{\bar{x}=\bar{a}}|:\bar{a}\in X(\bar{x})\}\}\subseteq\{\{|\mathbb{X}_{\bar{y}=\bar{a}}|:\bar{a}\in X(\bar{y})\}\}$,
    \item $\mathbb{X}\models \bar{x}\approx \bar{y}$ iff $|\mathbb{X}_{\bar{x}=\bar{a}}|=|\mathbb{X}_{\bar{y}=\bar{a}}|$ for all $\bar{a}\in A^{|\bar{x}|}$,
    \item $\mathbb{X}\models \bar{x}\approx^* \bar{y}$ iff $\{\{|\mathbb{X}_{\bar{x}=\bar{a}}|:\bar{a}\in X(\bar{x})\}\}=\{\{|\mathbb{X}_{\bar{y}=\bar{a}}|:\bar{a}\in X(\bar{y})\}\}$,
    \item $\mathbb{X}\models\bar{y}\perp_{\bar{x}}\bar{z}$ iff $|\mathbb{X}_{\bar{x}\bar{y}\bar{z}=\bar{a}\bar{b}\bar{c}}|\times|\mathbb{X}_{\bar{x}=\bar{a}}|=|\mathbb{X}_{\bar{x}\bar{y}=\bar{a}\bar{b}}|\times|\mathbb{X}_{\bar{x}\bar{z}=\bar{a}\bar{c}}|$ for all $\bar{a}\bar{b}\bar{c}\in A^{|\bar{x}\bar{y}\bar{z}|}$.
\end{itemize}
If $\mathbb{X}\models\tau$, we say that $\mathbb{X}$ \emph{satisfies} $\tau$. If $\Sigma$ is a set of atoms, we write $\mathbb{X}\models\Sigma$ iff $\mathbb{X}\models\sigma$ for all $\sigma\in\Sigma$. The next example demonstrates the satisfaction requirements for different dependencies in the general setting of a semiring $K$ with at least three nonzero elements. 
\begin{table}
\centering
        \begin{tabular}{c c c c c | c }
        \hline
        $ \ x \ $ & $ \ y \ $ & $ \ z \ $ & $ \ v \ $ & $ \ w \ $ & $ \ \mathbb{X} \ $  \\
        \hline
        0 & 0 & 1 & 0 & 0 & k  \\
        0 & 1 & 1 & 0 & 0 & h  \\
        1 & 0 & 0 & 1 & 0 & k'  \\
        \hline 
        \ ~ \
    \end{tabular}
\caption{The $K$-team $\mathbb{X}$ of Examples  \ref{satexample}--\ref{probabex} that is also used in the proof of Prop. \ref{addcan_prop}.}
\label{ex_table}
\end{table}
\begin{example}\label{satexample}
    Consider the $K$-team $\mathbb{X}$ of $A=\{0,1\}$ over $D=\{x,y,z,v,w\}$ depicted in Table \ref{ex_table}, and let $k,k',h$ be nonzero elements of $K$ such that $k+h< k'$. 
    First note that since $k,k',h$ are nonzero, the support of $\mathbb{X}$ consists of all of the three rows in the table. We have $\mathbb{X}\models\dep(x,z)$ because the values of $z$ agree on those rows where the values of $x$ agree. However, $\mathbb{X}\not\models\dep(x,y)$ because on the first two rows $x$ has the same value, 0, but the values for $y$ are different, namely 0 and 1. We also have $\mathbb{X}\models\dep(w)$ and $\mathbb{X}\not\models\dep(z)$ because the value of $w$ is constant, but the value of $z$ is not. 

We have $\mathbb{X}\models x\leq v$ and $\mathbb{X}\models x\approx v$ because $|\mathbb{X}_{x=0}|=k+h=|\mathbb{X}_{v=0}|$ and $|\mathbb{X}_{x=1}|=k'=|\mathbb{X}_{v=1}|$, but $\mathbb{X}\not\models x\leq z$ and $\mathbb{X}\not\models x\approx z$ because $|\mathbb{X}_{x=1}|=k' > k+h=|\mathbb{X}_{z=1}|$.
We also have $\mathbb{X}\models x\subseteq^* z$ and $\mathbb{X}\models x\approx^* z$ because $|\mathbb{X}_{x=0}|= k+h=|\mathbb{X}_{z=1}|$ and $|\mathbb{X}_{x=1}|=k'=|\mathbb{X}_{z=0}|$, so $\{{|\mathbb{X}_{x=a}|:a\in\{0,1\}\}}=\{\{k+h,k'\}\}=\{{|\mathbb{X}_{z=a}|:a\in\{0,1\}\}}$. However, $\mathbb{X}\not\models x\subseteq^* w$ and $\mathbb{X}\not\models x\approx^* w$ because $|\mathbb{X}_{x=0}|= k+h$, $|\mathbb{X}_{x=1}|=k'$, $|\mathbb{X}_{w=0}|= k+h+k'$, and $|\mathbb{X}_{w=1}|=0$, and hence $\{{|\mathbb{X}_{x=a}|:a\in\{0,1\}\}}=\{\{k+h,k'\}\}\not\subseteq\{\{k+h+k'\}\}=\{{|\mathbb{X}_{w=a}|:a\in\{0\}\}}$.

We have $\mathbb{X}\models x\perp w$ and $\mathbb{X}\not\models x\perp_w y$. The former holds since $|\mathbb{X}_{xw=ab}|\times|\mathbb{X}|=|\mathbb{X}_{xw=ab}|\times(k+h+k')=|\mathbb{X}_{x=a}|\times(k+h+k')=|\mathbb{X}_{x=a}|\times|\mathbb{X}_{w=j}|$ if $b=0$ and $|\mathbb{X}_{xw=ab}|\times|\mathbb{X}|=0\times(k+h+k')=|\mathbb{X}_{x=a}|\times 0=|\mathbb{X}_{x=a}|\times|\mathbb{X}_{w=b}|$ if $b=1$, implying that $|\mathbb{X}_{xw=ab}|\times|\mathbb{X}|=|\mathbb{X}_{x=a}|\times|\mathbb{X}_{w=b}|$ for all $a,b\in\{0,1\}$. The latter holds since $|\mathbb{X}_{wxy=011}|\times|\mathbb{X}_{w=0}|=0\neq k'\times h=|\mathbb{X}_{wx=01}|\times|\mathbb{X}_{wy=01}|$. 
\end{example}

Recall that we only consider positive semirings $K$; one reason for this is because otherwise functional dependency atoms lose the so-called \emph{locality} property that says that the satisfaction of an atom depends only on the (marginals of) variables that appear in the atom, i.e., an atom $\sigma$ is local if for all $\Var(\sigma)\subseteq V\subseteq D$, $\X\restriction V\models\sigma$ iff $\X\models\sigma$.  
\begin{example}\label{notpos}
    Consider the semiring of integral numbers $\mathbb{Z}=(\mathbb{Z},+,\times,0,1)$ that is not positive. Define the $\mathbb{Z}$-team $\mathbb{X}$ of $A=\{0,1\}$ over $D=\{x,y,z,v,w\}$ as depicted in Table \ref{ex_table}, and let $k=k'=1$ and $h=-1$. Then for $V=\{v,w\}$, the weight of the first two tuples collapses to 0 in $\X\restriction V$, and hence the FD $\dep(w,v)$ is not local, since $\X\restriction V\models\dep(w,v)$, but $\X\not\models\dep(w,v)$.
\end{example}
The nonlocality problem for nonpositive semirings demonstrated in the above example could be solved by integrating the requirement of locality in the definition of FD satisfaction in the semiring setting, i.e, defining that $\X\models\dep(\bar{x},\bar{y})^{*}$ iff for all $s,s'\in\Supp(\X\restriction \Var(\bar{x}\bar{y}))$, $s(\bar{x})=s'(\bar{x})$ implies $s(\bar{y})=s'(\bar{y})$. However, this approach would require us to abandon a desirable property called the \emph{transitivity} of functional dependencies: $\X\models\dep(\bar{x},\bar{y})^{*}$ and $\X\models\dep(\bar{y},\bar{z})^{*}$ would no longer imply $\X\models\dep(\bar{x},\bar{z})^{*}$. (See Example \ref{nonpos2} in Appendix \ref{exapp}.)
Positivity is also necessary for the soundness of some key axioms of Section \ref{axioms}, and other relevant properties, including the correspondence of conditional independence and lossless join decompositions \cite{hannula2023conditionalindependencesemiringrelations}. Hence, the assumption of positivity seems to be important in our setting. The following remark discusses some of the other assumptions that are relevant for some atoms.
\begin{remark}\label{remark}
    In the beginning of this section we also assumed that $K$ is commutative, totally zero-min ordered, and multiplicatively cancellative. A total zero-min order is needed for the atoms $\bar{x}\leq \bar{y}$, to ensure that the values that do not appear for $x$ in the support of $\mathbb{X}$ cannot prevent the satisfaction of the atom. Commutativity and multiplicative cancellativeness is needed for the atoms $\bar{y}\perp_{\bar{x}}\bar{z}$ to ensure the soundness of the usual axioms for (conditional) independence. For the atoms whose satisfaction condition does not refer to the multiplication of the semiring, it would also suffice to consider commutative monoids $(K,+,0)$. Since we consider these atoms together with (conditional) independence atoms that concern multiplication, for the sake of simplicity, we have decided to use only the framework of semirings. 
\end{remark}

Next, we consider two examples of semirings in which atoms can be interpreted as dependency notions that appear in databases, statistics, and probability theory. The following example illustrates how dependencies are interpreted in the Boolean semiring $\mathbb{B}$, and to which database dependencies they correspond.
\begin{example}
    If $K=\mathbb{B}$ with the usual order, i.e., $0<1$, then $\dep(\bar{x},\bar{y})$,  $\bar{x}\leq\bar{y}$, and $\bar{y}\perp_{\bar{x}}\bar{z}$ correspond to functional dependency $\Var(\bar{x})\rightarrow\Var(\bar{y})$, inclusion dependency $\bar{x}\subseteq\bar{y}$, and embedded multivalued dependency (EMVD) $\Var(\bar{x})\twoheadrightarrow\Var(\bar{y})|\Var(\bar{z})$ over (uni)relational databases. If $\bar{y}\perp_{\bar{x}}\bar{z}$ is an SCI, then it corresponds to multivalued dependency (MVD) $\Var(\bar{x})\twoheadrightarrow\Var(\bar{y})$. For any $\mathbb{B}$-team $\mathbb{X}$, the corresponding database relation is $\Supp(\mathbb{X})$.

    Consider the $K$-team $\mathbb{X}$ of $A=\{0,1\}$ over $D=\{x,y,z,v,w\}$ depicted in Table \ref{ex_table} in the case $K=\mathbb{B}$, and let $k=k'=h=1$. (Note that these values $k,k',h$ are different from Example \ref{satexample} as it is not the case that $k+h<k'$.) Functional dependencies are interpreted in the usual way in the support of $\mathbb{X}$ as explained in Example \ref{satexample}. We also have $\mathbb{X}\models x\leq z$ and $\mathbb{X}\not\models x\leq w$ because the values of $x$ are included in the values of $z$, but not in the values of $w$, i.e., the corresponding claims for inclusion dependencies hold in the database relation that consists of $\Supp(\mathbb{X})$. Lastly, we have $\mathbb{X}\models v\perp w$ and $\mathbb{X}\not\models xy\perp_w zv$ because $X(v)\times X(w)=X(vw)$ and $X(xy)\times X(zv)\neq X(xyzv)$ (note that the value of $w$ is constant), where $\times$ refers to the Cartesian product of the projections, i.e., the  claims for the corresponding EMVD, MVD, and the database relation hold.
\end{example}
\begin{remark}
    For notational convenience, we will sometimes write SCIs in the manner commonly used for multivalued dependencies (MVDs), that is, we denote by $\bar{x}\twoheadrightarrow\bar{y}$ the saturated conditional independence $\bar{y}\perp_{\bar{x}}\bar{z}$, where $\Var(\bar{x}\bar{y}\bar{z})=D$.
\end{remark}
Now, we consider the interpretations of dependencies for the semiring $\mathbb{R}_{\geq0}$, i.e., in the case of a probability distribution.
\begin{example}\label{probabex}
If $K=\mathbb{R}_{\geq0}$, then $\bar{x}\approx\bar{y}$, $\bar{x}\approx^*\bar{y}$, and $\bar{y}\perp_{\bar{x}}\bar{z}$ correspond to marginal identity, marginal distribution equivalence, and conditional independence over probability distributions. Note that each $\mathbb{R}_{\geq0}$-team $\mathbb{X}$ can be viewed as the probability distribution $p_{\mathbb{X}}$ obtained by normalizing $\mathbb{X}$, i.e., $p_{\mathbb{X}}\colon X\to[0,1]$ such that $p_{\mathbb{X}}(s)=\mathbb{X}(s)/|\mathbb{X}|$ for any $s\in X$.

Consider again the $K$-team $\mathbb{X}$ of $A=\{0,1\}$ over $D=\{x,y,z,v,w\}$ depicted in Table \ref{ex_table}, now in the case $K=\mathbb{R}_{\geq0}$, and let $k=1/2$ and $h=k'=1/4$. (Note that again it is not the case that $k+h<k'$.) Now $\mathbb{X}\models x\approx y$ and $\mathbb{X}\not\models x\approx z$ because the (marginal) probability distributions for $x$ and $y$ are the same, but different from the one for $z$. On the other hand, we have $\mathbb{X}\models x\approx^* z$ and $\mathbb{X}\not\models^* z\approx w$ because the multisets of marginal probabilities for relevant variables and values are (not) the same, i.e., $\{\{|\mathbb{X}_{x=a}|:a\in X(x)\}\}=\{\{|\mathbb{X}_{z=a}|:a\in X(z)\}\}=\{\{1/4,3/4\}\}\neq\{\{1\}\}=\{\{|\mathbb{X}_{w=a}|:a\in X(w)\}\}$. Lastly, we have $\mathbb{X}\models x\perp w$ and $\mathbb{X}\not\models x\perp_w y$ because the variables $x$ and $w$ are probabilistically independent, but the variables $x$ and $y$ conditioned on $w$ are not.
    \end{example}
    
As seen in the examples above, most of the atoms that we consider correspond to previously studied dependency notions from the literature for a suitable semiring. This seems not to be the case for the marginal distribution inclusion atom $\bar{x}\subseteq^*\bar{y}$. The idea of $\bar{x}\subseteq^*\bar{y}$ is to introduce a non-symmetric variant of the marginal distribution equivalence atom $\bar{x}\approx^*\bar{y}$ analogously to the way that $\bar{x}\leq\bar{y}$ relates to $\bar{x}\approx\bar{y}$. For example, in the case of the Boolean semiring, $\bar{x}\subseteq^*\bar{y}$ states that $|X(\bar{x})|\leq|X(\bar{y})|$, i.e., in the support of the team, the number of values for $\bar{x}$ is at most the number of values for $\bar{y}$. For the tropical semiring, $\bar{x}\subseteq^*\bar{y}$ additionally states that the minimal weight of any value in $X(\bar{x})$ must also appear in $X(\bar{y})$ for some value. 

\begin{remark}\label{othervariant}
One could also define a similar atom $\bar{x}\leq^* \bar{y}$ with a satisfaction condition involving the order of the semiring, namely that $\mathbb{X}\models \bar{x}\leq^* \bar{y}$ iff there is an injective function $f\colon X(\bar{x})\to X(\bar{y})$ such that for all $\bar{a}\in X(\bar{x})$,  we have  $|\mathbb{X}_{\bar{x}=\bar{a}}|\leq |\mathbb{X}_{\bar{y}=f(\bar{a})}|$. This would mean that $\bar{x}\leq\bar{y}$ always implies $\bar{x}\leq^*\bar{y}$, so 
the atom would allow a tighter analogy between the pair $\bar{x}\leq\bar{y}$ and $\bar{x}\leq^*\bar{y}$ and the pair  $\bar{x}\approx\bar{y}$ and $\bar{x}\approx^*\bar{y}$. In this work, we only study the ``subset'' variant $\bar{x}\subseteq^*\bar{y}$. Note that $\bar{x}\subseteq^*\bar{y}$ implies $\bar{x}\leq^*\bar{y}$, because the satisfaction condition of the former is equivalent to the following: there is an injective function $f\colon X(\bar{x})\to X(\bar{y})$ such that for all $\bar{a}\in X(\bar{x})$, we have $|\mathbb{X}_{\bar{x}=\bar{a}}|= |\mathbb{X}_{\bar{y}=f(\bar{a})}|$. Moreover, we can show that $\bar{x}\leq^*\bar{y}$ and $\bar{y}\leq^*\bar{x}$ imply $\bar{x}\approx^*\bar{y}$ by an argument similar to the proof of Proposition \ref{soundness}.
\end{remark}

In the next section, we will see that $\bar{x}\leq\bar{y}$ and $\bar{x}\subseteq^*\bar{y}$ are symmetric over additively cancellative semirings, so it does not make sense to consider these variants, e.g., in the case of semirings $\mathbb{N}$ or $\mathbb{R}_{\geq0}$.
Before that, we consider the next example of the (weighted) inclusion dependency in the case of a non-additively cancellative semiring that is not just the usual inclusion atom in the Boolean semiring.
\begin{example}\label{tropical}
    Let $K=\mathbb{T}$ with the reverse order of the extended reals. Then the atom $\bar{x}\leq\bar{y}$ indicates that if we wish to decrease (in the sense of the usual order of the reals) the marginal weight for each possible tuple of values appearing for $\bar{x}$, then choosing the tuple $\bar{y}$ is at least as good an option as the original tuple $\bar{x}$, if we do not mind possibly including new values to which the original weights might be partially distributed, as long as no values are removed.
    
    As a concrete example, consider a $K$-team containing information about possible allocations of some people to different work groups and the time needed for a core task that every group has to perform. Table \ref{tab:tasks} depicts a $K$-team $\mathbb{X}$ that demonstrates this idea. Each row in $\mathbb{X}$ corresponds to an available person and the weight is the amount of time that it takes for the person to finish a specific task that is performed by only one person per group at the time. The variables $x$, $y$, $z$, and $u$ correspond to some possible allocations of the people to groups that work together. We have $\mathbb{X}\models x\leq y$ because $x\leq y$ means that moving from the allocation $x$ to the allocation $y$ does not remove any of the existing groups and does not increase the time that the quickest person in each group needs for the task. On the other hand, we have $\mathbb{X}\not\models x\leq z$ because the group C is missing from the allocation $z$, and we have $\mathbb{X}\not\models x\leq u$ because the minimum time in group A is increased from 0.75 hours to 1 hour when we move from $x$ to $u$.
\end{example}

\begin{table}[]
    \centering
    \begin{tabular}{cccc|c}
    \hline
       $x$  & $y$ & $z$ &$u$ & $\mathbb{X}$  \\
       \hline
       Group A  & Group D & Group A  & Group B  &  0.75 \\
       Group B & Group C & Group B  & Group B  & 0.5 \\
       Group B & Group A  & Group B  & Group B  & 0.5 \\
       Group B        & Group B &  Group B  & Group A  &  1  \\
       Group C & Group B & Group A  & Group C  &  0.5 \\
       \hline
    \end{tabular}
    \caption{Table depicts the $K$-team $\mathbb{X}$ described in Example \ref{tropical}, where each row corresponds to a person, $x$ and $y$ are possible allocations of the people into groups, and the $K$-team $\mathbb{X}$ records the amount of hours it takes for each person to finish the task. We have $\mathbb{X}\models x\leq y$, but $\mathbb{X}\not\models x\leq z$ and $\mathbb{X}\not\models x\leq u$.}
    \label{tab:tasks}
\end{table}

\section{Implication Problems and Axioms under Semiring Semantics}

In this section, we define implication problems and axiomatizations in the general framework of semiring semantics. We will see later in the next section that certain axiomatizations completely characterize implication problems over suitable classes of semirings.

\subsection{Implication Problems}

For a set of atoms $\Sigma\cup\{\tau\}$  such that $\Var(\Sigma\cup\{\tau\})\subseteq D$, we write
$\Sigma\models_K\tau$ if $\mathbb{X}\models\Sigma$ implies $\mathbb{X}\models\tau$ for all $K$-teams $\mathbb{X}$ of $X$ over $D$. If $\Sigma\models_K\tau$, we say that $\Sigma$ \emph{logically implies} $\tau$ over $K$. The implication problem for a class of atoms $\mathcal{C}$ over $K$ is defined as follows: given a set $\Sigma\cup\{\tau\}$ of atoms from class $\mathcal{C}$, decide whether $\Sigma\models_K\tau$. If the semiring $K$ is clear from the context, we write $\Sigma\models\tau$ instead of $\Sigma\models_K\tau$.

Before we define axiomatizations to characterize the implication problems for different classes of atoms, we present some results concerning logical implication that depend on certain properties of the semiring.
For these results, we use the following lemma, whose proof is in Appendix \ref{exapp}.
\begin{lemma}\label{addcan_lemma}
 If $K$ is additively cancellative, then $k<k'$ and $h\leq h'$ imply $k+h<k'+h'$.   
\end{lemma}
\begin{proposition}\label{addcan_prop}
\begin{itemize}
    \item[(i)] 
   $\{\bar{x}\leq \bar{y}\}\models_K\bar{y}\leq \bar{x}$  iff $K$ is additively cancellative\footnote{In \cite{hannula2026dichotomyaxiomatisinginclusiondependencies}, the respective characterization for symmetry is formulated in terms of weak additive cancellativity instead of additive cancellativity. This difference stems from the fact that the authors of \cite{hannula2026dichotomyaxiomatisinginclusiondependencies} consider only the natural order, so for all $\bar{a}\in A^{|\bar{x}|}$, $|\mathbb{X}_{\bar{x}=\bar{a}}|\leq|\mathbb{X}_{\bar{y}=\bar{a}}|$ implies that there is $k_{\bar{a}}\in K$ such that $|\mathbb{X}_{\bar{x}=\bar{a}}|+k_{\bar{a}}=|\mathbb{X}_{\bar{y}=\bar{a}}|$, and hence $\sum_{\bar{a}\in A^{|\bar{x}|}}|\mathbb{X}_{\bar{x}=\bar{a}}|+\sum_{\bar{a}\in A^{|\bar{x}|}}k_{\bar{a}}=\sum_{\bar{a}\in A^{|\bar{y}|}}|\mathbb{X}_{\bar{y}=\bar{a}}|=\sum_{\bar{a}\in A^{|\bar{x}|}}|\mathbb{X}_{\bar{x}=\bar{a}}|$, which means that $k_{\bar{a}}=0$ for all $\bar{a}\in A^{|\bar{x}|}$. If the order is not the natural order, such $k_{\bar{a}}$ might not exist. See Appendix \ref{exapp}, Proposition \ref{doubletwo} for an example of a semiring that is weakly additively cancellative, but not additively cancellative. Note that a weakly additively cancellative semiring that is totally ordered by the natural order is also additively cancellative. (Appendix \ref{exapp}, Proposition \ref{naturordprop}.)}, 
    \item[(ii)] $\{\bar{x}\subseteq^* \bar{y}\}\models_K\bar{y}\subseteq^* \bar{x}$ iff $K$ weakly additively cancellative.
\end{itemize}
\end{proposition} 
\begin{proof}
We first show the right-to-left directions of the claims. For item (i), we show that if $K$ is additively cancellative, then $|\mathbb{X}_{\bar{x}=\bar{a}}|\leq|\mathbb{X}_{\bar{y}=\bar{a}}|$ for all $\bar{a}\in A^{|\bar{x}|}$ implies $|\mathbb{X}_{\bar{x}=\bar{a}}|=|\mathbb{X}_{\bar{y}=\bar{a}}|$ for all $\bar{a}\in A^{|\bar{x}|}$, which means that $\{\bar{x}\leq \bar{y}\}\models_K\{\bar{y}\leq \bar{x}\}$. Suppose for a contradiction that $|\mathbb{X}_{\bar{x}=\bar{a}}|<|\mathbb{X}_{\bar{y}=\bar{a}}|$ for some $\bar{a}\in A^{|\bar{x}|}$. Then by Lemma \ref{addcan_lemma}, $\sum_{\bar{a}\in A^{|\bar{x}|}}|\mathbb{X}_{\bar{x}=\bar{a}}|<\sum_{\bar{a}\in A^{|\bar{x}|}}|\mathbb{X}_{\bar{y}=\bar{a}}|$. This is a contradiction, as $\sum_{\bar{a}\in A^{|\bar{x}|}}|\mathbb{X}_{\bar{x}=\bar{a}}|=|\mathbb{X}|=\sum_{\bar{a}\in A^{|\bar{y}|}}|\mathbb{X}_{\bar{y}=\bar{a}}|$.

For item (ii), we show that if $K$ is weakly additively cancellative, 
then $\{\{|\mathbb{X}_{\bar{x}=\bar{a}}|:\bar{a}\in X(\bar{x})\}\}\subseteq\{\{|\mathbb{X}_{\bar{y}=\bar{a}}|:\bar{a}\in X(\bar{y})\}\}$ implies $\{\{|\mathbb{X}_{\bar{x}=\bar{a}}|:\bar{a}\in X(\bar{x})\}\}=\{\{|\mathbb{X}_{\bar{y}=\bar{a}}|:\bar{a}\in X(\bar{y})\}\}$, which means that $\{\bar{x}\subseteq^* \bar{y}\}\models_K\{\bar{y}\subseteq^* \bar{x}\}$.
    From 
    $\{\{|\mathbb{X}_{\bar{x}=\bar{a}}|:\bar{a}\in X(\bar{x})\}\}\subseteq\{\{|\mathbb{X}_{\bar{y}=\bar{a}}|:\bar{a}\in X(\bar{y})\}\}$, it follows that there exists an injective function $f\colon X(\bar{x})\to X(\bar{y})$ such that $|\mathbb{X}_{\bar{x}=\bar{a}}|=|\mathbb{X}_{\bar{y}=f(\bar{a})}|$ for all $\bar{a}\in X(\bar{x})$. It suffices to show that $f$ is surjective. If this is not the case, let $B=X(\bar{y})\setminus f(X(\bar{x}))$. Now we have $|\mathbb{X}|=\sum_{\bar{a}\in X(\bar{x})}|\mathbb{X}_{\bar{x}=\bar{a}}|=\sum_{\bar{a}\in X(\bar{x})}|\mathbb{X}_{\bar{y}=f(\bar{a})}|\neq\sum_{\bar{a}\in X(\bar{x})}|\mathbb{X}_{\bar{y}=f(\bar{a})}|+\sum_{\bar{b}\in B}|\mathbb{X}_{\bar{y}=\bar{b}}|=\sum_{\bar{a}\in X(\bar{y})}|\mathbb{X}_{\bar{y}=\bar{a}}|=|\mathbb{X}|$, where the inequality follows from $0\neq\sum_{\bar{b}\in B}|\mathbb{X}_{\bar{y}=\bar{b}}|$ by the assumption. Since $|\mathbb{X}|\neq|\mathbb{X}|$ is clearly a contradiction, the claim follows.

 We now show the left-to-right directions by contraposition. Without loss of generality, we show the claims for unary atoms. For item (i), suppose that $K$ is not additively cancellative. Then there exists $k,k',h\in K$ such that $k+h=k+k'$, but $h< k'$. Let $\mathbb{X}$ be as in Table \ref{ex_table} of Example \ref{satexample}, and assume that the elements $k,k',h\in K$ of the table are such that $k+h=k+k'$, but $h< k'$. Then $|\mathbb{X}_{x=0}|=k+h=k+k'=|\mathbb{X}_{y=0}|$ and $|\mathbb{X}_{x=1}|=h<k'=|\mathbb{X}_{y=1}|$, so $\mathbb{X}\models x\leq y$, but $\mathbb{X}\not\models y\leq x$, meaning that $\{\bar{x}\leq \bar{y}\}\not\models_K\{\bar{y}\leq \bar{x}\}$. 
 
 For item (ii), suppose that $K$ is weakly absorptive. Let again $\mathbb{X}$ be as in Table \ref{ex_table}, but this time assume that $k,k',h\in K$ are such that $k+h=k$, $k,h\neq 0$, and $k'=0$. Then $|\mathbb{X}_{x=0}|=k+h=k=k+k'=|\mathbb{X}_{y=0}|$, but $|\mathbb{X}_{x=1}|=k'=0\neq h=|\mathbb{X}_{y=1}|$. Hence, $\{\{|\mathbb{X}_{x=a}|:a\in X(x)\}\}=\{\{k\}\}\subsetneq\{\{k,h\}\}=\{\{|\mathbb{X}_{y=a}|:a\in X(y)\}\}$, so $\mathbb{X}\models x\subseteq^* y$, but $\mathbb{X}\not\models y\subseteq^* x$, meaning that $\{\bar{x}\subseteq^* \bar{y}\}\not\models_K\{\bar{y}\subseteq^* \bar{x}\}$.
\end{proof} 
As a consequence of Proposition \ref{addcan_prop},  $\mathbb{X}\models\bar{x}\leq\bar{y}$ iff $\mathbb{X}\models\bar{x}\approx\bar{y}$ and $\mathbb{X}\models\bar{x}\subseteq^*\bar{y}$ iff $\mathbb{X}\models\bar{x}\approx^*\bar{y}$ for any $K$-team $\mathbb{X}$, where $K$ is additively cancellative.
Hence, for additively cancellative semirings, we consider the atoms $\bar{x}\approx\bar{y}$ and $\bar{x}\approx^*\bar{y}$ 
instead of $\bar{x}\leq\bar{y}$, 
and $\bar{x}\subseteq^*\bar{y}$, 
and forget the order of $K$.

\subsection{Sound Axioms}\label{axioms}

An axiom is an expression of the form $\Lambda\vdash\tau$, where $\Lambda$ is a set of atoms and $\tau$ is a single atom. For a set of axioms $\mathcal{A}$, we define the deduction relation $\vdash_{\mathcal{A}}$ such that $\Sigma\vdash_{\mathcal{A}}\sigma$ if and only if there is a finite sequence $(\tau_1,\dots,\tau_n)$ of atoms such that $\tau_n=\sigma$, and for each $\tau_i$, $i\in\{1,\dots,n\}$, either $\tau_i\in\Sigma$ or $\Lambda\vdash\tau_i\in\mathcal{A}$ for some $\Lambda\subseteq\Sigma\cup\{\tau_1,\dots,\tau_{i-1}\}$. If $\Sigma\vdash_{\mathcal{A}}\sigma$, we say that there is an $\mathcal{A}$-deduction of $\sigma$ from $\Sigma$. Note that if $\Lambda\vdash\tau\in\mathcal{A}$, then also $\Lambda\vdash_{\mathcal{A}}\tau$. We sometimes write $\Sigma\vdash\sigma$ instead of $\Sigma\vdash_{\mathcal{A}}\sigma$, if the set of axioms $\mathcal{A}$ is clear from the context. We say that the set of axioms $\mathcal{A}$ is \emph{sound} for a class of semirings $\mathcal{C}$, if $\Sigma\vdash_{\mathcal{A}}\tau$ implies $\Sigma\models_K\tau$ for all $\Sigma\cup\{\tau\}$ and all $K\in\mathcal{C}$.

We denote by $\mathcal{A}_{FD}$ the \emph{Armstrong axiomatization} \cite{armstrong74} consisting of the following three axioms on the left;  reflexivity, transitivity, and
augmentation. When restricted to unary FDs, we use the notation $\mathcal{A}_{UFD}$ for the set of the following axioms on the right.
\begin{multicols}{2}
\begin{enumerate}[leftmargin=0.85cm]
    \item[FD1] $\emptyset\vdash \dep(\bar{x},\bar{x}'), \Var(\bar{x}')\subseteq\Var(\bar{x})$
    \item[FD2] $\{\dep(\bar{x},\bar{y}),\dep(\bar{y},\bar{z})\}\vdash \dep(\bar{x},\bar{z})$
    \item[FD3] $\{\dep(\bar{x},\bar{y})\}\vdash \dep(\bar{x}\bar{z},\bar{y}\bar{z})$
\end{enumerate}
    \columnbreak 
\begin{enumerate}[leftmargin=1.5cm]
    \item[UFD1] $\emptyset\vdash \dep(x,x)$
    \item[UFD2] $\{\dep(x,y),\dep(y,z)\}\vdash \dep(x,z)$
    \item[UFD3] $\{\dep(y)\}\vdash \dep(x,y)$
\end{enumerate}
\end{multicols}
\noindent Let $1\leq k\leq n$, and denote by $\rho_k$ any function $\rho_k\colon D^n\to D^k$ such that $\rho_k(x_1,\dots,x_n)=(x_{i_1},\dots,x_{i_k})$ for all $(x_1,\dots,x_n)\in D^n$, where $i_1,\dots,i_k$ is a sequence of distinct elements from $\{1,\dots, n\}$. Let $P_k^n$ be the set of these functions $\rho_k$, and define $P^n=\bigcup_{1\leq k\leq n} P_k^n$. We define the axiomatizations $\mathcal{A}_{IND}$ and $\mathcal{A}_{MI}$,
consisting of the following axioms for INDs \cite{casanova82} and MIs \cite{HANNULA22}, respectively.
\begin{multicols}{2}
    \begin{enumerate}[leftmargin=1cm]
    \item[IND1] $\emptyset\vdash \bar{x}\leq \bar{x}$
    \item[IND2] $\{\bar{x}\leq \bar{y},\bar{y}\leq \bar{z}\}\vdash \bar{x}\leq\bar{z}$
    \item[IND3] $\{\bar{x}\leq \bar{y}\}\vdash {\rho(\bar{x})\leq \rho(\bar{y}), \rho\in P^{|\bar{x}|}}$
    \end{enumerate}
    \columnbreak 
    \begin{enumerate}[leftmargin=1.25cm]
    \item[MI1] $\emptyset\vdash \bar{x}\approx \bar{x}$
    \item[MI2] $\{\bar{x}\approx \bar{y}\}\vdash \bar{y}\approx \bar{x}$
    \item[MI3] $\{\bar{x}\approx \bar{y},\bar{y}\approx \bar{z}\}\vdash \bar{x}\approx \bar{z}$
    \item[MI4] $\{\bar{x}\approx \bar{y}\}\vdash {\rho(\bar{x})\approx \rho(\bar{y}), \rho\in P^{|\bar{x}|}}$
    \end{enumerate}
\end{multicols}
\noindent For MDIs and MDEs, we define the sets $\mathcal{A}_{MDI}$ and $\mathcal{A}_{MDE}$, consisting of the corresponding axioms below.
    \begin{enumerate}[leftmargin=1.2cm]
    \item[MDI1] $\emptyset\vdash \bar{x}\subseteq^* \bar{x}$
    \item[MDI2] $\{\bar{x}\subseteq^* \bar{y},\bar{y}\subseteq^* \bar{z}\}\vdash \bar{x}\subseteq^* \bar{z}$
    \item[MDI3] $\{\bar{x}\subseteq^* \bar{y}\}\vdash {\rho(\bar{x})\subseteq^* \rho'(\bar{y}), \rho\in P_{|\bar{x}|}^{|\bar{x}|},\rho'\in P_{|\bar{y}|}^{|\bar{y}|}}$.
    \end{enumerate}
    \begin{enumerate}[leftmargin=1.2cm]
    \item[MDE1] $\emptyset\vdash \bar{x}\approx^* \bar{x}$
    \item[MDE2] $\{\bar{x}\approx^* \bar{y}\}\vdash \bar{y}\approx^* \bar{x}$
    \item[MDE3] $\{\bar{x}\approx^* \bar{y},\bar{y}\approx^* \bar{z}\}\vdash \bar{x}\approx^* \bar{z}$
    \item[MDE4] $\{\bar{x}\approx^* \bar{y}\}\vdash {\rho(\bar{x})\approx^* \rho'(\bar{y}), \rho\in P_{|\bar{x}|}^{|\bar{x}|},\rho'\in P_{|\bar{y}|}^{|\bar{y}|}}$.
    \end{enumerate}
For UINDs, UMDIs, UMIs, and UMDEs, we define the axiomatizations $\mathcal{A}_{UIND}$, $\mathcal{A}_{UMDI}$, $\mathcal{A}_{UMI}$, and $\mathcal{A}_{UMDE}$, consisting of the unary variants of the axioms IND1-2, MDI1-2, MI1-3, and MDE1-3, respectively.
For CIs \cite{pearl88}, IAs \cite{geiger:1991,hannula2021interactionfunctionalinclusiondependencies}, and SCIs \cite{10.1145/509404.509414}, we define the axioms $\mathcal{A}_{CI}$, $\mathcal{A}_{IA}$,and $\mathcal{A}_{SCI}$, consisting of the axioms listed below, respectively.
\begin{enumerate}[leftmargin=0.8cm]
     \item[CI1] $\emptyset\vdash\emptyset\perp_{\bar{x}}\bar{y}$
    \item[CI2] $\{\bar{y}\perp_{\bar{x}}\bar{z}\}\vdash\bar{z}\perp_{\bar{x}}\bar{y}$
    \item[CI3] $\{\bar{y}\perp_{\bar{x}}\bar{z}\}\vdash\bar{y}'\perp_{\bar{x}}\bar{z}',\Var(\bar{y}')\subseteq\Var(\bar{y}), \Var(\bar{z}')\subseteq\Var(\bar{z})$ \item[CI4] $\{\bar{y}\perp_{\bar{x}}\bar{z}\bar{w}\}\vdash\bar{y}\perp_{\bar{x}\bar{w}}\bar{z}$ \item[CI5] $\{\bar{y}\perp_{\bar{x}}\bar{z},\bar{y}\perp_{\bar{x}\bar{z}}\bar{w}\}\vdash\bar{y}\perp_{\bar{x}}\bar{z}\bar{w}$ 
    \item[CI6] $\{\bar{y}\perp_{\bar{x}}\bar{z},\bar{w}\perp_{\bar{x}}\bar{w}\}\vdash\bar{y}\perp_{\bar{x}}\bar{z}\bar{w}$ 
\end{enumerate}
\begin{enumerate}[leftmargin=0.85cm]
     \item[IA1] $\emptyset\vdash\emptyset\perp\bar{x}$
    \item[IA2] $\{\bar{x}\perp\bar{y}\}\vdash\bar{y}\perp\bar{x}$
    \item[IA3] $\{\bar{x}\perp\bar{y}\}\vdash\bar{x}'\perp\bar{y}', \Var(\bar{x}')\subseteq\Var(\bar{x}), \Var(\bar{y}')\subseteq\Var(\bar{y})$ 
    \item[IA4] $\{\bar{x}\perp\bar{y},\bar{x}\bar{y}\perp\bar{z}\}\vdash\bar{x}\perp\bar{y}\bar{z}$
    \item[IA5] $\{\bar{x}\perp\bar{y},\bar{z}\perp\bar{z}\}\vdash\bar{x}\perp\bar{y}\bar{z}$
\end{enumerate}
    \begin{enumerate}[leftmargin=1cm]
    \item[SCI1]  $\{\bar{x}\twoheadrightarrow\bar{y}\}\vdash \bar{x}\twoheadrightarrow\bar{z},\Var(\bar{y})\cap\Var(\bar{z})\subseteq\Var(\bar{x}),\Var(\bar{x}\bar{y}\bar{z})=D$
    \item[SCI2] $\emptyset\vdash \bar{x}\twoheadrightarrow\bar{y},\Var(\bar{y})\subseteq\Var(\bar{x})$
    \item[SCI3] $\{\bar{x}\twoheadrightarrow\bar{y}\}\vdash \bar{x}\bar{w}\twoheadrightarrow\bar{y}\bar{z},\Var(\bar{z})\subseteq\Var(\bar{w})$
    \item[SCI4]  $\{\bar{x}\twoheadrightarrow\bar{y},\bar{y}\twoheadrightarrow\bar{z}\}\vdash \bar{x}\twoheadrightarrow\bar{w}, \Var(\bar{w})=\Var(z)\setminus\Var(\bar{y})$
    \end{enumerate}
For the interaction of MIs and MDEs \cite{hirvonen24}, MDEs and IAs \cite{hirvonen24conf}, FDs and IAs \cite{hannula2021interactionfunctionalinclusiondependencies}, and FDs and SCIs \cite{10.1145/509404.509414}, respectively, we define the following axioms.
\begin{enumerate}[leftmargin=1.95cm]
    \item[MI \& MDE] $\{\bar{x}\approx \bar{y}\}\vdash\bar{x}\approx^* \bar{y}$
\end{enumerate}
\begin{multicols}{2}
\begin{enumerate}[leftmargin=2.15cm]
    \item[MDE \& IA 1] $\{\bar{x}\approx^* \bar{y},\bar{y}\perp\bar{y}\}\vdash\bar{x}\perp\bar{x}$
    \item[MDE \& IA 2] $\{\bar{x}\perp\bar{x},\bar{y}\perp\bar{y}\}\vdash\bar{x}\approx^* \bar{y}$
\end{enumerate}
\begin{enumerate}[leftmargin=1.95cm]
    \item[FD \& IA 1] $\{\bar{x}\perp \bar{x},\}\vdash \dep(\bar{x})$
    \item[FD \& IA 2] $\{\bar{x}\perp \bar{y},\dep(\bar{x},\bar{y})\}\vdash \bar{y}\perp\bar{y}$
\end{enumerate}
\end{multicols}
\begin{enumerate}[leftmargin=2cm]
    \item[FD \& SCI 1] $\{\dep(\bar{x},\bar{y})\}\vdash\bar{x}\twoheadrightarrow\bar{y}$
    \item[FD \& SCI 2] $\{\bar{x}\twoheadrightarrow\bar{z},\dep(\bar{y}\setminus\bar{z},\bar{z}')\}\vdash\dep(\bar{x},\bar{w}),\Var(\bar{w})=\Var(\bar{z})\cap\Var(\bar{z}')$
\end{enumerate}
We also define the following notation for the sets of axioms for MIs and MDEs: $\mathcal{A}_{MI \& MDE}:=\{MI \ \& \ MDE \}$, for IAs and MDEs: $\mathcal{A}_{MDE \& IA}:=\{MDE \ \& \ IA \ 1,\\ MDE \ \& \ IA \ 2\}$, for IAs and FDs: $\mathcal{A}_{FD \& IA}:=\{FD \ \& \ IA \ 1,FD \ \& \ IA \  2\}$, and for SCIs and FDs: $\mathcal{A}_{FD \& SCI}:=\{FD \ \& \ SCI \ 1,FD \ \& \ SCI \  2\}$. We also use analogous notation for the unary variants of the axioms and sets defined above.

Let $\ell\in\{1,3,5,\dots\}$ and define $S_\ell\colon\{0,1,\dots,\ell\}\to\{0,1,\dots,\ell\}$ such that $S_\ell(i)= i+1 \mod \ell+1$. 
Then we define the following axioms for the interaction of UFDs and UINDs \cite{CosmadakisKV90}, UFDs and UMIs  \cite{hirvonen24}, and UFDs and UMDEs  \cite{hirvonen24}, respectively.
\begin{enumerate}[leftmargin=0cm]
\item[] $\ell$-CYCLE-UIND
    \item[] $\{\dep(x_{i},x_{S_\ell(i)}):0\leq i\leq \ell-1, i\text{ even}\}\cup\{x_{S_\ell(i)}\leq x_{i}, 1\leq i\leq \ell, i\text{ odd}\}\vdash$\\
    ${\dep(x_{S_\ell(i)},x_{i})},0\leq i\leq \ell-1, i\text{ even, and } x_{i}\leq x_{S_\ell(i)}, 1\leq i\leq \ell, i\text{ odd}$ 
\end{enumerate}
\begin{enumerate}[leftmargin=0cm]
\item[] $\ell$-CYCLE-UMI
    \item[] $\{\dep(x_{i},x_{S_\ell(i)}):0\leq i\leq \ell-1, i\text{ even}\}\cup\{x_{i}\approx x_{S_\ell(i)}:1\leq i\leq \ell, i\text{ odd}\}\vdash{\dep(x_{S_\ell(i)},x_{i}),0\leq i\leq \ell-1, i\text{ even}}$ 
\end{enumerate}
\begin{enumerate}[leftmargin=0cm]
\item[] $\ell$-CYCLE-UMDE
    \item[] $\{\dep(x_{i},x_{S_\ell(i)}):0\leq i\leq \ell-1, i\text{ even}\}\cup\{x_{i}\approx^*x_{S_\ell(i)}:1\leq i\leq \ell, i\text{ odd}\}\vdash$\\
    ${\dep(x_{S_\ell(i)},x_{i}),0\leq i\leq \ell-1, i\text{ even, and} x_{i}\approx^*x_{S_\ell(i)},0\leq i\leq \ell-1, i\text{ even}}$ 
\end{enumerate}
Since the cycle rules are defined for each odd $\ell$, we define the following infinite sets of axioms to cover all $\ell$-cycle rules $\mathcal{A}_{CYCLE-UIND}:=\{\ell\text{-CYCLE-UIND}:\ell\in\mathbb{N}, \ell\text{ is odd}\}$, $\mathcal{A}_{CYCLE-UMI}:=\{\ell\text{-CYCLE-UMI}:\ell\in\mathbb{N}, \ell\text{ is odd}\}$, and $\mathcal{A}_{CYCLE-UMDE}:=\{\ell\text{-CYCLE-UMDE}:\ell\in\mathbb{N}, \ell\text{ is odd}\}$.

We now show that the above axioms are sound for suitable classes of semirings.
\begin{proposition}\label{soundness}
    The following sets of axioms are sound for any semiring: $\mathcal{A}_{FD}$, $\mathcal{A}_{IND}$, $\mathcal{A}_{MDI}$, $\mathcal{A}_{MI}$, $\mathcal{A}_{MDE}$, $\mathcal{A}_{MI \& MDE}$, $\mathcal{A}_{MDE \& IA}$, $\mathcal{A}_{CYCLE-UIND}$, $\mathcal{A}_{CYCLE-UMI}$, and $\mathcal{A}_{CYCLE-UMDE}$.
\end{proposition}
\begin{proof}
    The axioms and the satisfaction relation for FDs are the same as in the case of relational databases \cite{armstrong74}, so it is easy to see that the axioms are sound. Excluding the cycle rules, the soundness of the other axioms follows straightforwardly from the definitions of the marginal and the satisfaction relation for the relevant atoms.
    
    For the cycle rules, the soundness proofs are similar to the Boolean and the probability semiring cases of \cite{CosmadakisKV90} and \cite{hirvonen24}, except that for $\ell$-cycle-UIND the situation is slightly more complicated in the case of general $K$. Note that since $\bar{x}\approx \bar{y}$ implies $\bar{x}\approx^* \bar{y}$, it suffices to show that $\ell$-CYCLE-UIND and  $\ell$-CYCLE-UMDE are sound.
    
    For a team $\mathbb{X}\colon X\to K$, we have that $\dep(x,y)$ implies $|X(y)|\leq |X(x)|$, $x\leq y$ implies $X(x)\subseteq X(y)$, and $x\approx^* y$ implies $|X(x)|=|X(y)|$. Hence, the existence of the alternating FD/IND (or FD/UMDE) cycle $x_0,x_1\dots,x_\ell,x_0$ implies that $|X(x_i)|=|X(x_j)|$ for all $i,j\in\{0,1,\dots,\ell\}$. By the finiteness of $X$, this means that each surjective function associated with $\dep(x_i,x_{S_\ell(i)})$, where $i$ even, is also injective, implying that $\dep(x_{S_\ell(i)},x_i)$. 

    For UMDEs, note that the existence of the bijective function $f$ associated with $\dep(x_i,x_{S_\ell(i)})$ and $\dep(x_{S_\ell(i)},x_i)$ implies that $|\mathbb{X}_{x_i=a}|=|\mathbb{X}_{x_{S_\ell(i)}=f(a)}|$ for all $a\in X(x_i)$ and $|\mathbb{X}_{x_i=f^{-1}(b)}|=|\mathbb{X}_{x_{S_\ell(i)}=b}|$ for all $b\in X(x_{S(i)})$, so we have $x_i\approx^* x_{S_\ell(i)}$. 

    For the case of UINDs, first define an equivalence relation $\sim$ on $A$ such that  $a\sim b$ iff $|\mathbb{X}_{x_0=a}|=|\mathbb{X}_{x_0=b}|$. Denote by $[a]$ the equivalence class of $a$ under $\sim$. The total zero-min order on $K$ induces an order on the equivalence classes such that $[a]\leq[b]$ iff $|\mathbb{X}_{x_0=a}|\leq|\mathbb{X}_{x_0=b}|$. Let $a_0\in A$ be such that its equivalence class is the smallest in this order. Then
    \begin{align*}
    |\mathbb{X}_{x_0=a_0}|&=|\mathbb{X}_{x_1=f_1(a_0)}|\geq|\mathbb{X}_{x_2=f_1(a_0)}|=|\mathbb{X}_{x_3=(f_3\circ f_1)(a_0)}|\geq|\mathbb{X}_{x_4=(f_3\circ f_1)(a_0)}|=\dots\geq|\mathbb{X}_{x_{\ell-1}=(f_{\ell-2}\circ\dots\circ f_1)(a_0)}|\\&
    =|\mathbb{X}_{x_\ell=(f_\ell\circ\dots\circ f_1)(a_0)}|\geq|\mathbb{X}_{x_0=(f_\ell\circ\dots\circ f_1)(a_0)}|,
    \end{align*}
where each $f_i\colon A\to A$ is some bijective extention of the bijective function from $X(x_{S_\ell^{-1}(i)})$ to $X(x_i)$ associated with $\dep(x_{S_\ell^{-1}(i)},x_{i})$ for an odd $i$. If any of the inequalities is strict, then $|\mathbb{X}_{x_0=a_0}|>|\mathbb{X}_{x_0=(f_\ell\circ\dots\circ f_1)(a_0)}|$. This is a contradiction as $|\mathbb{X}_{x_0=a_0}|$ is the smallest marginal of $x_0$. This means that all of the inequalities above are actually equalities, i.e, $|\mathbb{X}_{x_i=(f_i\circ\dots\circ f_1)(a_0)}|=|\mathbb{X}_{x_{S_\ell(i)}=(f_i\circ\dots\circ f_1)(a_0)}|$ for all odd $i$.
Moreover, since each $f_i$ is injective, the composite function $f_\ell\circ\dots\circ f_1$ is also injective, and $(f_\ell\circ\dots\circ f_1)\restriction[a_0]$ must be a permutation on $[a_0]$. 

Suppose then that $|\mathbb{X}_{x_i=(f_i\circ\dots\circ f_1)(a_{n})}|=|\mathbb{X}_{x_{S_\ell(i)}=(f_i\circ\dots\circ f_1)(a_{n})}|$ for all odd $i$ and the function $(f_\ell\circ\dots\circ f_1)\restriction[a_{n}]$ is a permutation on $[a_{n}]$.
Let then $a_{n+1}\in A$ be such that its equivalence class is minimal in $(A/\sim)\setminus\{[a_0],\dots,[a_{n}]\}$. Now $(f_\ell\circ\dots\circ f_1)(a_{n+1})\notin \bigcup_{j=0}^{n}[a_j]$, and we can repeat an argument analogous to the case of $a_0$ to obtain that $|\mathbb{X}_{x_i=(f_i\circ\dots\circ f_1)(a_{n+1})}|=|\mathbb{X}_{x_{S_\ell(i)}=(f_i\circ\dots\circ f_1)(a_{n+1})}|$ for all odd $i$ and the function $(f_\ell\circ\dots\circ f_1)\restriction[a_{n+1}]$ is a permutation on $[a_{n+1}]$.
Hence, we have $|\mathbb{X}_{x_i=(f_i\circ\dots\circ f_1)(a)}|=|\mathbb{X}_{x_{S_\ell(i)}=(f_i\circ\dots\circ f_1)(a)}|$ for all $a\in A$ and odd $i$.
Since $(f_i\circ\dots\circ f_1)(A)=A$, we have
$|\mathbb{X}_{x_i=a}|=|\mathbb{X}_{x_{S_\ell(i)}=a}|$ for all $a\in A$, and $x_i\leq x_{S_\ell(i)}$ for all odd $i$.
\end{proof}
The soundness of $\mathcal{A}_{CI}$, $\mathcal{A}_{SCI}$, and $\mathcal{A}_{FD \& SCI}$ over commutative and multiplicatively cancellative semirings was already proven in \cite{hannula2023conditionalindependencesemiringrelations} for disjoint CIs; dropping the disjointness assumption does not essentially change the proofs. The soundness of CI6, which explicitly involves a nondisjoint atom, is also easy to see. (Cf. the case of IA5 in the proof of Proposition \ref{soundness2}.) 
\begin{proposition}[\cite{hannula2023conditionalindependencesemiringrelations}]
    The axioms $\mathcal{A}_{CI}$, $\mathcal{A}_{SCI}$, and $\mathcal{A}_{FD \& SCI}$ are sound for any commutative and multiplicatively cancellative semiring.
\end{proposition}
The following claim is straightforward to check, so we only include a short description of the proof idea.  
\begin{proposition}\label{soundness2}
    The axioms $\mathcal{A}_{IA}$ and $\mathcal{A}_{FD \& IA}$ are sound for any commutative and multiplicatively cancellative semiring.
\end{proposition}
\begin{proof}
    For IAs, it is easy to see that the soundness of IA1 and IA2 follows from the multiplicative cancellativity and commutativity, respectively. The soundness of IA3 follows from the definition of the marginal using distributivity. The soundness of IA4 follows from the multiplicative cancellativity, commutativity, and the soundness of IA3. By considering the satisfaction requirement of $\bar{z}\perp\bar{z}$, we notice that $\bar{z}\perp\bar{z}$ implies that the value of $\bar{z}$ has to be constant in the support of the team. This means that the soundness of the axiom IA5 is easy to verify by considering the satisfaction requirement of the consequent $\bar{x}\perp\bar{y}\bar{z}$ separately for each value of $\bar{z}$.

    Since $\dep(\bar{x})$ also states that the value of $\bar{x}$ is constant in the support of the team, the soundness of the axiom FD \& IA 1 is easy to see. 
    For FD \& IA 2, assume that $\mathbb{X}\not\models\dep(\bar{y})$. Let $\bar{a},\bar{b}\in X(\bar{y})$ be such that $\bar{a}\neq \bar{b}$.  Since $\mathbb{X}\models\dep(\bar{x},\bar{y})$, there is a surjective function $f\colon X(\bar{x})\to X(\bar{y})$. Now we have $f^{-1}(\bar{a}),f^{-1}(\bar{b})\in X(\bar{x})$ such that $f^{-1}(\bar{a})\neq f^{-1}(\bar{b})$. Since $\mathbb{X}\models \bar{x}\perp \bar{y}$, we have $|\mathbb{X}_{\bar{x}\bar{y}=f^{-1}(\bar{a})\bar{b}}|\times|\mathbb{X}|=|\mathbb{X}_{\bar{x}=f^{-1}(\bar{a})}|\times|\mathbb{X}_{\bar{y}=\bar{b}}|$. This is a contradiction because $|\mathbb{X}_{\bar{x}\bar{y}=f^{-1}(\bar{a})\bar{b}}|\times|\mathbb{X}|=0\times|\mathbb{X}|=0$, but $|\mathbb{X}_{\bar{x}=f^{-1}(\bar{a})}|\neq 0$ and $|\mathbb{X}_{\bar{y}=\bar{b}}|\neq 0$.
\end{proof}

\begin{figure}
    \centering
    \scalebox{0.7}{
\begin{tikzpicture}[
    every node/.style={font=\small, 
    inner sep=2pt},
    arrow/.style={->, thick}
]

\node (sci) at (0,0) {SCI};
\node (fd2sci) at (2,0) {FD\&SCI};
\node (fd) at (11,-2) {FD};

\node (fd_sci) at (1,1.5) {FD+SCI};

\node (cycle_uind) at (3.25,3.5) {CYCLE-UIND};
\node (uind) at (3.5,4.5) {UIND};

\node (ufd) at (3.3,1.5) {UFD};
\node (ufd_and_ia) at (4.5,1.2) {UFD\&IA};
\node (ia) at (6,1.3) {IA};
\node (ufd_ia) at (4.5,2.4) {UFD+IA};

\node (umi) at (7.6,1.5) {UMI};
\node (umde) at (8,0.5) {UMDE};
\node (umiandumde) at (9,-0.25) {UMI\&UMDE};
\node (umi_umde) at (9.2,2.5) {UMI+UMDE};
\node (ia_umde) at (7.9,4) {IA\&UMDE};

\node (cycle_umde) at (10,-1.2) {CYCLE-UMDE};

\node (fd_uind_sci) at (1,6.5) {FD+UIND+SCI};
\node (ufd_uind_ia) at (5,6.5) {UFD+UIND+IA};
\node (umi_umde_ia) at (9,6.5) {UMI+UMDE+IA};
\node (fd_umi_umde) at (13,6.5) {FD+UMI+UMDE};

\draw[arrow] (sci.north) to[out=100,in=250] (fd_sci.south west);
\draw[arrow] (fd2sci.north) to[out=100,in=260] (fd_sci.south);
\draw[arrow] (fd.south west) to[out=200,in=350] (fd_sci.south east);

\draw[arrow] (fd_sci.north) to[out=90,in=260] (fd_uind_sci.south west);

\draw[arrow] (uind.north west) to[out=90,in=260]  (fd_uind_sci.south east);
\draw[arrow] (cycle_uind.north west) to[out=90,in=260] (fd_uind_sci.south);

\draw[arrow] (uind.north east) to[out=90,in=260] (ufd_uind_ia.south west);
\draw[arrow] (cycle_uind.north east) to[out=90,in=260]  (ufd_uind_ia.south);

\draw[arrow] (ufd.north) to[out=90,in=200]  (ufd_ia.south west);
\draw[arrow] (ufd_and_ia.north) to[out=90,in=260]  (ufd_ia.south);
\draw[arrow] (ia.north west) to[out=90,in=290]  (ufd_ia.south east);

\draw[arrow] (ufd_ia.north) to[out=60,in=260]  (ufd_uind_ia.south east);

\draw[arrow] (umi.north) to[out=90,in=230]  (umi_umde.south west);
\draw[arrow] (umde.north) to[out=90,in=260]  (umi_umde.south);
\draw[arrow] (umiandumde.north) to[out=60,in=300]  (umi_umde.south east);

\draw[arrow] (ia.north east) to[out=100,in=250] (umi_umde_ia.south west);

\draw[arrow] (umi_umde.north) to[out=70,in=260] (umi_umde_ia.south east);
\draw[arrow] (ia_umde.north) to[out=80,in=260] (umi_umde_ia.south);

\draw[arrow] (fd.north east) to[out=50,in=290] (fd_umi_umde.south east);
\draw[arrow] (umi_umde.north east) to[out=90,in=250] (fd_umi_umde.south west);
\draw[arrow] (cycle_umde.north) to[out=40,in=290] (fd_umi_umde.south);

\end{tikzpicture}}
    \caption{Depiction of the construction of the axiomatizations. In the figure, the name of dependency class $\mathcal{C}$ refers to the set of axioms $\mathcal{A}_{\mathcal{C}}$, and the arrows refer to set inclusion. For example, $\mathcal{A}_{FD+UIND+SCI}=\mathcal{A}_{FD+SCI}\cup\mathcal{A}_{UIND}\cup\mathcal{A}_{CYCLE-UIND}$, where $\mathcal{A}_{FD+SCI}=\mathcal{A}_{FD}\cup\mathcal{A}_{SCI}\cup\mathcal{A}_{FD\&SCI}$. The axiomatizations $\mathcal{A}_{FD+UMI+SCI}$ and $\mathcal{A}_{FD+UMI+IA}$ are not depicted, but can be obtained by changing $\mathcal{A}_{UIND}$ and $\mathcal{A}_{CYCLE-UIND}$ to $\mathcal{A}_{UMI}$ and $\mathcal{A}_{CYCLE-UMI}$, respectively.}
    \label{axiomt_pic}
\end{figure}
\begin{remark}
    In the naming of the axiom sets in this paper, ``$\&$'' refers to the interaction axioms, and ``$+$'' to the axioms for joint classes. For example, $\mathcal{A}_{FD\&SCI}=\{FD \ \& \ SCI \ 1,FD \ \& \ SCI \  2\}$, which contains the axioms that characterize the interaction of FDs and SCIs, and $\mathcal{A}_{FD+SCI}=\mathcal{A}_{FD}\cup\mathcal{A}_{SCI}\cup\mathcal{A}_{FD\&SCI}$, which contains the axioms for the joint class of FDs and SCIs. See Figure \ref{axiomt_pic} for a depiction of the construction of the axiomatizations.
\end{remark}
\section{Complete Axiomatizations}

In this section, we show that some of the axioms from the previous section form axiomatizations that completely characterize certain implication problems, i.e., they are complete for the corresponding class of atoms over a suitable class of semirings. We say that axiomatization $\mathcal{A}$ is \emph{complete} for a class of atoms $\mathcal{C}$ over a class of semirings $\mathcal{C}'$, if $\Sigma\models_K\tau$ implies $\Sigma\vdash_{\mathcal{A}}\tau$ for all sets $\Sigma\cup\{\tau\}$ of atoms from the class $\mathcal{C}$ and all semirings $K$ from $\mathcal{C}'$.

We begin by noting that some of the axiomatizations from Section \ref{axioms} are already known to be complete in the general setting of semirings. The axiomatization for MIs over probability distributions introduced in \cite{HANNULA22} is known to be complete for any positive semiring.
\begin{theorem}[\cite{hannula2023conditionalindependencesemiringrelations}]
    The axiomatization $\mathcal{A}_{MI}$ is sound and complete for the implication problem for MIs over any semiring.
\end{theorem}
Let $\mathcal{A}_{FD+SCI}:=\mathcal{A}_{FD}\cup\mathcal{A}_{SCI}\cup\mathcal{A}_{FD\& SCI}$.
The following result is known for the implication problem of FDs+SCIs.
\begin{theorem}[\cite{hannula2023conditionalindependencesemiringrelations}]
    The axiomatization $\mathcal{A}_{FD+SCI}$ is sound and complete for the implication problem for FDs+SCIs over any multiplicatively cancellative and commutative semiring.
\end{theorem}
We now show that many axiomatizations originally introduced in the relational or probabilistic setting, i.e, $K=\mathbb{B}$ or $K=\mathbb{R}_{\geq0}$, are complete more generally over some classes of semirings. Recall that throughout this paper we assume that the semiring $K$ is positive. Since the soundness of the axiomatizations follows from the results of Section \ref{axioms}, it suffices to prove the completeness for each theorem. 

Separate considerations for semirings with and without additive cancellativity are important for our results because additive cancellativity affects the symmetry properties of some atoms. Recall that if $K$ is additively cancellative, the atoms $\bar{x}\leq\bar{y}$ and $\bar{x}\approx\bar{y}$, and the atoms $\bar{x}\subseteq^*\bar{y}$ and $\bar{x}\approx^*\bar{y}$ are equivalent, so for additively cancellative semirings, we only consider the symmetric variants of the atoms. If $K$ is not additively cancellative, we can also consider the non-symmetric variants $\bar{x}\leq\bar{y}$ and $\bar{x}\subseteq^*\bar{y}$.

\subsection{Classes of semirings with additive cancellativity}

We first present some results for additively cancellative semirings. Examples of additively cancellative semirings include the semiring of natural numbers, the probability semiring, and the semiring of multivariate polynomials with natural number coefficients. We consider the implication problems for FDs+UMIs+UMDEs and UMIs+UMDEs+IAs. The assumption of additive cancellativity is natural, when one considers UMIs and UMDEs since it is related to the symmetry properties discussed in Proposition \ref{addcan_prop}. The assumption will also turn out to be useful for generalizing the proof methods from the setting of probability distributions to semiring teams, as it allows us to use similar counterexample teams in the completeness proofs.

We first show that the axiomatization
$\mathcal{A}_{FD+UMI+UMDE}:=\mathcal{A}_{FD}\cup\mathcal{A}_{UMI}\cup\mathcal{A}_{UMDE}\cup\mathcal{A}_{UMI \& UMDE} \cup\mathcal{A}_{CYCLE-UMDE}$
is sound and complete for the implication problem for FDs+UMIs+UMDEs over any additively cancellative semiring $K$. The proof is a simple modification of the proof for the analogous completeness result in the more restricted case of probability distributions \cite{hirvonen24}. 

Let $\mathcal{A}$ be an axiomatization and $D$ a set of variables such that $\Var(\Sigma)\subseteq D$. Define the closure of $\Sigma$ under $\mathcal{A}$ and $D$ as $\Cl_{\mathcal{A},D}(\Sigma):=\{\sigma:\Sigma\vdash_{\mathcal{A}}\sigma,\Var(\sigma)\subseteq D\}$.
If $\mathcal{A}$ and $D$ are clear from the context, we just write $\Cl(\Sigma)$. 
The definitions and lemmas from \cite{hirvonen24} that describe the construction of the set $X$ that is used in Lemma \ref{FDUMIUMDE_lemma} are included in Appendix \ref{definitions} (Definitions \ref{graph_def} and \ref{addcan_def}, Lemma \ref{graph_prop_lemma}). The properties listed in the lemma follow from \cite{hirvonen24}.

\begin{lemma}\label{FDUMIUMDE_lemma}
    Let $\Sigma$ be a set of FDs, UMIs, and UMDEs such that $\Var(\Sigma)\subseteq D$. Then there exists $A=\{0,\dots,n\}$, and a set $X$ of assignments $s\colon D\to A$ that satisfies the following properties, where $(A_u,m_u):=\{\{s(u)\in A:s\in X\}\}$ for all $u\in D$. 
    \begin{itemize}
        \item[(i)] for $\sigma:=\dep(\bar{x},\bar{y})$, $\sigma\in\Cl(\Sigma)$ iff for all $s,s'\in X$, $s(\bar{x})=s'(\bar{x})$ implies $s(\bar{y})=s'(\bar{y})$,
       \item[(ii)] for $\sigma:=x\approx y$, $\sigma\in\Cl(\Sigma)$ iff $(A_x,m_x)=(A_y,m_y)$,
       \item[(iii)] for $\sigma:=x\approx^* y$, $\sigma\in\Cl(\Sigma)$ iff $\{\{m_x(i):i\in A_x\}\}=\{\{m_y(i):i\in A_y\}\}$.
    \end{itemize}
\end{lemma}

\begin{theorem}\label{FDUMIUMDEthrm}
    The axiomatization $\mathcal{A}_{FD+UMI+UMDE}$ is sound and complete for the implication problem for FDs+\ UMIs+UMDEs over any additively cancellative semiring.
\end{theorem}
\begin{proof} 
Construct the set $X$ as in Lemma \ref{FDUMIUMDE_lemma}, and define $\mathbb{X}\colon X\to K$ as $\mathbb{X}(s)=1$ for all $s\in X$. We show that $\mathbb{X}\models\sigma$ iff $\sigma\in\Cl(\Sigma)$.
Assume first that $\sigma:=\dep(\bar{x},\bar{y})$. Note that $\Supp(\mathbb{X})=X$, so by Lemma \ref{FDUMIUMDE_lemma}, item (i), we have $\mathbb{X}\models\sigma$ iff $\sigma\in\Sigma$.

Assume then that $\sigma:=x\approx y$. Extend the functions $m_u:A_u\to\{1,2,\dots\}$ to $A$ by defining $m_u(i)=0$ for any $i\in A\setminus A_u$. Note that by positivity and additive cancellativity of $K$, $\sum_{1\leq j\leq m_x(i)}1\neq\sum_{1\leq j\leq m_y(i)}1$ for any $i\in A$ such that $m_x(i)\neq m_y(i)$. By Lemma \ref{FDUMIUMDE_lemma}, item (ii), we have $\sigma\in\Cl(\Sigma)$ iff $(A_x,m_x)=(A_y,m_y)$. By the definition of $\mathbb{X}$, we then have $(A_x,m_x)=(A_y,m_y)$ iff $|\mathbb{X}_{x=i}|=\sum_{1\leq j\leq m_x(i)}1=\sum_{1\leq j\leq m_y(i)}1=|\mathbb{X}_{y=i}|$ for all $i\in A$. Hence, $\mathbb{X}\models\sigma$ iff $\sigma\in\Cl(\Sigma)$.

Assume finally that $\sigma:=x\approx^* y$. Then by Lemma \ref{FDUMIUMDE_lemma}, item (iii), $\sigma\in\Cl(\Sigma)$ iff $\{\{m_x(i):i\in A_x\}\}=\{\{m_y(i):i\in A_y\}\}$ iff $\{\{|\mathbb{X}_{x=i}|:i\in X(x)\}\}=\{\{|\mathbb{X}_{y=i}|:i\in X(y)\}\}$ iff $\mathbb{X}\models\sigma$. 
\end{proof}
Note that the proof of Theorem \ref{FDUMIUMDEthrm} also entails the existence of the so-called \emph{Armstrong relations} for the class FD+UMI+UMDE.
\begin{definition}
    Let $\Sigma$ be a set of atoms of class $\mathcal{C}$. We say that a $K$-team $\mathbb{X}$ is an \emph{Armstrong relation} for $\Sigma$ over $K$ if $\mathbb{X}\models\sigma$ iff $\Sigma\models\sigma$. We say that class $\mathcal{C}$ \emph{enjoys} \emph{Armstrong relations} over $K$ if every finite set of atoms of class $\mathcal{C}$ has an Armstrong relation over $K$.
\end{definition}
\begin{corollary}\label{armstrongcor}
    The class FD+UMI+UMDE enjoys Armstrong relations over any positive and additively cancellative semiring.
\end{corollary}
It is known that for probability distributions (i.e., $K=\mathbb{R}_{\geq0}$), the only interaction between IAs and UMIs or UMDEs is via constancy atoms expressed as $x\perp x$ \cite{hirvonen24conf}. 
We define $\mathcal{A}_{UMI+UMDE+IA}:=\mathcal{A}_{UMI}\cup\mathcal{A}_{UMDE}\cup\mathcal{A}_{IA}\cup\mathcal{A}_{UMI \& UMDE}\cup\mathcal{A}_{IA\& UMDE}$. The proof of the following lemma is essentially the completeness proof for UMIs, UMDEs, and IAs over probability distributions from \cite{hirvonen24conf}. We include it in Appendix \ref{proof} for the convenience of the reader.
\begin{lemma}[\cite{hirvonen24conf}]\label{IAUMIUMDE_lemma}
        Let $\Sigma\cup\{\tau\}$ be a set of 
        UMIs, UMDEs, and IAs 
         such that $\Var(\Sigma\cup\{\tau\})\subseteq D$ and $\tau\notin\Cl(\Sigma)$. Then there exists 
         a $\mathbb{R}_{\geq0}$-team $\mathbb{X}:X\to\mathbb{R}_{\geq0}$ over $D$ such that $\mathbb{X}(s)=1/|X|$ for all $s\in X$, and $\mathbb{X}\models\Sigma$, but $\mathbb{X}\not\models\tau$.
\end{lemma}
Now we obtain the following theorem using the above lemma. 
\begin{theorem}\label{IAUMIUMDE_thm}
    The axiomatization $\mathcal{A}_{UMI+UMDE+IA}$ is sound and complete for the implication problem for UMIs+\ UMDEs+IAs over any additively and multiplicatively cancellative commutative semiring.
\end{theorem}
\begin{proof}
     Suppose that $\Sigma\not\vdash\tau$. We show that $\Sigma\not\models\tau$ by constructing a $K$-team $\mathbb{X}$ such that $\mathbb{X}\models\Sigma$, but $\mathbb{X}\not\models\tau$.
Construct the set $X$ as in Lemma \ref{IAUMIUMDE_lemma}, and define $\mathbb{X}\colon X\to K$ as $\mathbb{X}(s)=1$ for all $s\in X$. Let $\mathbb{X}'$ be the $\mathbb{R}_{\geq0}$-team of Lemma \ref{IAUMIUMDE_lemma}. We show that for all $\sigma\in\Sigma\cup\{\tau\}$,  $\mathbb{X}\models\sigma$ iff $\mathbb{X}'\models\sigma$.
    Now $\mathbb{X}'\models\bar{x}\perp\bar{y}$ iff $|\mathbb{X}'_{\bar{x}\bar{y}=\bar{a}\bar{b}}|\times|\mathbb{X}'|=|\mathbb{X}'_{\bar{x}=\bar{a}}|\times|\mathbb{X}'_{\bar{y}=\bar{b}}|$ for all $\bar{a}\bar{b}\in A^{|\bar{x}\bar{y}|}$. Note that by the definition of $\mathbb{X}$, we have $|\mathbb{X}_{\bar{z}=\bar{c}}|=\sum_{1\leq i\leq |X|\times|\mathbb{X}'_{\bar{z}=\bar{c}}|}1$ for all $\bar{z}\in D^{|\bar{z}|}$ and $\bar{c}\in A^{|\bar{z}|}$. Then $|\mathbb{X}'_{\bar{x}\bar{y}=\bar{a}\bar{b}}|\times|\mathbb{X}'|=|\mathbb{X}'_{\bar{x}=\bar{a}}|\times|\mathbb{X}'_{\bar{y}=\bar{b}}|$ iff $|X|\times|\mathbb{X}'_{\bar{x}\bar{y}=\bar{a}\bar{b}}|\times|X|\times|\mathbb{X}'|=|X|\times|\mathbb{X}'_{\bar{x}=\bar{a}}|\times|X|\times|\mathbb{X}'_{\bar{y}=\bar{b}}|$ iff $|\mathbb{X}_{\bar{x}\bar{y}=\bar{a}\bar{b}}|\times|\mathbb{X}|=|\mathbb{X}_{\bar{x}=\bar{a}}|\times|\mathbb{X}_{\bar{y}=\bar{b}}|$. Therefore, $\mathbb{X}'\models\bar{x}\perp\bar{y}$ iff $\mathbb{X}\models\bar{x}\perp\bar{y}$. The claim for UMIs and UMDEs follows analogously from the observation that $|\mathbb{X}_{\bar{z}=\bar{c}}|=\sum_{1\leq i\leq |X|\times|\mathbb{X}'_{\bar{z}=\bar{c}}|}1$ for all $\bar{z}\in D^{|\bar{z}|}$ and $\bar{c}\in A^{|\bar{z}|}$. 
   
\end{proof}

\begin{remark}\label{generalization}
Additive cancellativity is strictly speaking not necessary for the results of this section: it is enough that the semiring has an additive submonoid that resembles the natural numbers. Theorem \ref{FDUMIUMDEthrm} and Corollary \ref{armstrongcor} (Theorem \ref{IAUMIUMDE_thm}) hold(s) more generally for any positive (multiplicatively cancellative and commutative) semiring that contains an element $k\in K$ for which $\sum_{1\leq i\leq n} k\neq\sum_{1\leq i\leq m}  k$ for all $n,m\in\mathbb{N}$ such that $n\neq m$, i.e., the submonoid $(\langle k\rangle,+,0)$) generated by $k$ is isomorphic to $(\mathbb{N},+,0)$ without an order. This element $k$ can be used, instead of element 1, to construct the $K$-teams in the proofs.   
\end{remark}

\subsection{Classes of semirings without additive cancellativity}\label{withoutaddcan}
We now move to the results for classes of semirings without additive cancellativity. The semiring classes relevant for each result are specified in the theorem statements. Examples of semirings for which the results of this section hold include the Boolean semiring, the tropical semiring, and the Viterbi semiring. Those results that do not require multiplicative cancellativity hold also for the fuzzy semiring.

For unary marginal distribution inclusion, the axiomatization $\mathcal{A}_{UMDI}$ is complete for weakly absorptive semirings.
Recall that for positive semiring $K$, $K$ is weakly absorptive iff $K$ is is not weakly additively cancellative, and by Proposition \ref{addcan_prop} (ii), weak additive cancellativity characterizes when the atom is symmetric.
\begin{theorem}\label{UINDstarthrm}
    The axiomatization $\mathcal{A}_{UMDI}$ is sound and complete for the implication problem for UMDIs over any weakly absorptive semiring. 
\end{theorem}
\begin{proof}
 Suppose that $\Sigma\not\vdash\tau$. We show that $\Sigma\not\models\tau$ by constructing a very simple $K$-team $\mathbb{X}$ such that $\mathbb{X}\models\Sigma$, but $\mathbb{X}\not\models\tau$.
  Let $\tau:=x\subseteq^*y$, and define $Z:=\{z\in D:\Sigma\vdash x\subseteq^*z\}$. Let $X:=\{s,s'\}$, where $s,s'\colon D\to\{0,1\}$ are such that $s(u)=0$ for all $u\in D$, and $s'(u)=1$ if $u\in Z$ and $s'(u)=0$ if $u\in  D\setminus Z$.
    Define then $\mathbb{X}\colon X\to K$ such that $\mathbb{X}(s)=a$ and $\mathbb{X}(s')=b$, where $a$ and $b$ are nonzero elements such that $a+b=a$ (which exist by the weak absorptivity of $K$). 
    
    Note that if $u\in Z$ and $u'\in D\setminus Z$, then $|\mathbb{X}_{u=0}|=a$, $|\mathbb{X}_{u=1}|=b$, $|\mathbb{X}_{u'=0}|=a+b=a$, and $|\mathbb{X}_{u'=1}|=0$. Therefore, $\mathbb{X}\not\models u\subseteq^* u'$ iff  $u\in Z$ and $u'\in D\setminus Z$. Now clearly $\mathbb{X}\not\models\tau$, because $x\in Z$ and $y\in D\setminus Z$. Let $v\subseteq^* w\in\Sigma$. Then either $\Sigma\vdash x\subseteq^* v$ or $\Sigma\not\vdash x\subseteq^* v$. In the first case, by UMDI2, we have $\Sigma\vdash x\subseteq^* w$, so $v,w\in Z$, and $\mathbb{X}\models v\subseteq^* w$. In the latter case, $v\in D\setminus Z$, and thus also $\mathbb{X}\models v\subseteq^* w$.
\end{proof}
\begin{remark}
    The class of UMDIs does not enjoy Armstrong relations over $\mathbb{B}$: by the definition of UMDI, for all $\mathbb{B}$-teams $\mathbb{X}$ over $D=\{x,y\}$, we have $\mathbb{X}\models x\subseteq^*y$ or $\mathbb{X}\models y\subseteq^*x$, but $\{x\subseteq^* x,y\subseteq^* y\}\not\models x\subseteq^*y$ and $\{x\subseteq^* x,y\subseteq^* y\}\not\models y\subseteq^*x$.
\end{remark}
Next, we show that the axiomatizations $\mathcal{A}_{IND}$,
$\mathcal{A}_{FD+UIND+SCI}:=\mathcal{A}_{FD}\cup\mathcal{A}_{UIND}\cup\mathcal{A}_{SCI}\cup\mathcal{A}_{SCI\& FD}\cup\mathcal{A}_{CYCLE-UIND}$, and $\mathcal{A}_{UFD+UIND+IA}:=\mathcal{A}_{UFD}\cup\mathcal{A}_{UIND}\cup\mathcal{A}_{IA}\cup\mathcal{A}_{FD \& IA}\cup\mathcal{A}_{CYCLE-UIND}$
are sound and complete for the respective implication problems for INDs, FDs+UINDs+SCIs, and UFDs+UINDs+IAs over any (commutative) totally zero-min ordered self absorptive semiring $K$. 
Note that the examples of non-additively cancellative semirings given in the preliminaries section 
are all self absorptive, so this class of semirings covers, e.g., the Boolean, tropical, and Viterbi semirings.

In the case of the Boolean semiring, the axiomatizations  $\mathcal{A}_{IND}$, $\mathcal{A}_{FD+UIND+SCI}$, and $\mathcal{A}_{UFD+UIND+IA}$  correspond to the sound and complete axiomatizations for the finite implication problems for inclusion dependencies \cite{casanova82}, for functional dependencies, unary inclusion dependencies, and saturated conditional independencies (MVDs) \cite{CosmadakisKV90}\footnote{In \cite{CosmadakisKV90}, Cosmadakis et al. show completeness more generally for unary inclusion dependencies and \textit{full implicational dependencies} (UINDs+FIDs) with respect to finite implication. Both FDs and MVDs are FIDs, so the result holds in particular for FDs+UINDs+SCIs.}, and for unary functional dependencies, unary inclusion dependencies, and independence atoms \cite{hannula2021interactionfunctionalinclusiondependencies} over (uni)relational databases. This means that we can utilize these results in the form of the following lemma to prove the completeness in the more general case of semirings.
 Note that in the lemma we assume that the Boolean semiring has the usual order, i.e., $0<1$. 
\begin{lemma}[\cite{casanova82,CosmadakisKV90,Fagin82,hannula2021interactionfunctionalinclusiondependencies}]\label{lemmafd}
\begin{enumerate}[label=(\roman*),ref=\ref{lemmafd} (\roman*)]
    \item\label{FDUIND_lemma}
    Let $\Sigma\cup\{\tau\}$ be a set of FDs, UINDs, and SCIs 
         such that $\Var(\Sigma\cup\{\tau\})\subseteq D$ and $\tau\notin\Cl(\Sigma)$. Then there exists  
         a Boolean team $\mathbb{X}:X\to\{0,1\}$ over $D$ such that $\mathbb{X}(s)=1$ for all $s\in X$, and $\mathbb{X}\models\Sigma$, but $\mathbb{X}\not\models\tau$.
    \item\label{FDUINDIA_lemma}
    Let $\Sigma$ be a set of INDs or a set of UFDs, UINDs, and IAs
         such that $\Var(\Sigma)\subseteq D$. Then there exists a Boolean team $\mathbb{X}:X\to\{0,1\}$ over $D$ such that $\mathbb{X}(s)=1$ for all $s\in X$, and $\mathbb{X}\models\sigma$ iff $\sigma\in\Cl(\Sigma)$ for all $\sigma$.
\end{enumerate}   
\end{lemma}
\begin{remark}\label{boolremark}
    For any Boolean team $\mathbb{X}:X\to\{0,1\}$, we have $\mathbb{X}\models \bar{x}\leq \bar{y}$ iff $X(\bar{x})\subseteq X(\bar{y})$, and $\mathbb{X}\models\bar{y}\perp_{\bar{x}}\bar{z}$ iff for all $s,s'\in X$ such that $s(\bar{x})=s'(\bar{x})$, there is $s''\in X$ such that $s''(\bar{x}\bar{y})=s(\bar{x}\bar{y})$ and $s''(\bar{x}\bar{z})=s'(\bar{x}\bar{z})$.
\end{remark}
\begin{theorem}\label{INDax}
    The axiomatization $\mathcal{A}_{IND}$ is sound and complete for the implication problem for INDs over any totally zero-min ordered self absorptive semiring. 
\end{theorem}
\begin{proof}
 Suppose that $\Sigma\not\vdash\tau$. We show that $\Sigma\not\models\tau$ by constructing a $K$-team $\mathbb{X}$ such that $\mathbb{X}\models\Sigma$, but $\mathbb{X}\not\models\tau$.
 Let the set $X$ of assignments $s\colon D\to A$ be as in Lemma \ref{FDUINDIA_lemma}. Let $k\in K$ be an idempotent nonzero element, and define $\mathbb{X}\colon X\to K$ as $\mathbb{X}(s)=k$ for all $s\in X$.
 
If $\sigma:=\bar{x}\leq \bar{y}\in\Sigma$, then by Lemma \ref{FDUINDIA_lemma} and Remark \ref{boolremark}, we have $X(\bar{x})\subseteq X(\bar{y})$. By the definition of $\mathbb{X}$, we then have $|\mathbb{X}_{\bar{x}=\bar{a}}|=0\leq|\mathbb{X}_{\bar{y}=\bar{a}}|$ or 
$|\mathbb{X}_{\bar{x}=\bar{a}}|=k=|\mathbb{X}_{\bar{y}=\bar{a}}|$ for all $\bar{a}\in A^{|\bar{x}|}$.
     Suppose then that $\tau:=\bar{x}\leq \bar{y}$. By Lemma \ref{FDUINDIA_lemma} and Remark \ref{boolremark}, there is $\bar{a}\in X(\bar{x})$ such that $\bar{a}\notin X(\bar{y})$. This means that $|\mathbb{X}_{\bar{x}=\bar{a}}|=k>0=|\mathbb{X}_{\bar{y}=\bar{a}}|$, so $\mathbb{X}\not\models\tau$.
\end{proof}
\begin{theorem}\label{FDUINDax}
    The axiomatization $\mathcal{A}_{FD+UIND+SCI}$ is sound and complete for the implication problem for FDs+\ UINDs+SCIs over any totally zero-min ordered multiplicatively cancellative and commutative idempotent semiring.
\end{theorem}
\begin{proof}
 Suppose that $\Sigma\not\vdash\tau$. We show that $\Sigma\not\models\tau$ by constructing a $K$-team $\mathbb{X}$ such that $\mathbb{X}\models\Sigma$, but $\mathbb{X}\not\models\tau$.
 Let the set $X$ of assignments $s\colon D\to A$ be as in Lemma \ref{FDUIND_lemma}, and define $\mathbb{X}$ analogously to the proof of Theorem \ref{INDax}.
Since $k\neq0$, we have $\Supp(\mathbb{X})=X$. Hence, by Lemma \ref{FDUIND_lemma}, the claims are trivial for $\sigma:=\dep(\bar{x},\bar{y})\in\Sigma$ and $\tau:=\dep(\bar{x},\bar{y})$. The proofs of the claims for $\sigma:=x\leq y\in\Sigma$ and $\tau:=x\leq y$ are the same as in the proof of Theorem \ref{INDax}.

 Note that by the definition of $\X$, we have $\X\models\bar{y}\perp_{\bar{x}}\bar{z}$ iff $|\mathbb{X}_{\bar{x}\bar{y}\bar{z}=\bar{a}\bar{b}\bar{c}}|\times|\mathbb{X}_{\bar{x}=\bar{a}}|=|\mathbb{X}_{\bar{x}\bar{y}=\bar{a}\bar{b}}|\times|\mathbb{X}_{\bar{x}\bar{z}=\bar{a}\bar{c}}|$ for all $\bar{a}\bar{b}\bar{c}\in A^{|\bar{x}\bar{y}\bar{z}|}$ iff for all $s,s'\in X$ such that $s(\bar{x})=s'(\bar{x})$, there is $s''\in X$ such that $s''(\bar{x}\bar{y})=s(\bar{x}\bar{y})$ and $s''(\bar{x}\bar{z})=s'(\bar{x}\bar{z})$.
 Then, by Lemma \ref{FDUIND_lemma} and Remark \ref{boolremark}, the claim holds for $\sigma:=\bar{y}\perp_{\bar{x}}\bar{z}\in\Sigma$ and $\tau:=\bar{y}\perp_{\bar{x}}\bar{z}$. 
\end{proof}
\begin{theorem}\label{UFDUINDIAthr}
    The axiomatization $\mathcal{A}_{UFD+UIND+IA}$ is sound and complete for the implication problem for UFDs+\ UINDs+IAs over any totally zero-min ordered multiplicatively cancellative and commutative idempotent semiring.
\end{theorem}
\begin{proof}
 We show that there is a $K$-team $\mathbb{X}$ such that $\mathbb{X}\models\sigma$ iff $\sigma\in\Cl(\Sigma)$.  Let the set $X$ of assignments $s\colon D\to A$ be as in Lemma \ref{FDUINDIA_lemma}, and define $\mathbb{X}$ analogously to the proof of Theorem \ref{INDax}. For $\sigma:=\dep(x,y)$ or $\sigma:=x\leq y$, the claims can be shown as in the proofs of Theorems \ref{INDax} and \ref{FDUINDax}. 
 
 Let $\sigma:=\bar{x}\perp \bar{y}$. Then by Lemma \ref{FDUINDIA_lemma} and Remark \ref{boolremark}, we have $\sigma\in\Cl(\Sigma)$ iff for all $s,s'\in X$, there is $s''\in X$ such that $s''(\bar{x})=s(\bar{x})$ and $s''(\bar{y})=s'(\bar{y})$ iff $|\mathbb{X}_{\bar{x}\bar{y}=\bar{a}\bar{b}}|\times|\mathbb{X}|=|\mathbb{X}_{\bar{x}=\bar{a}}|\times|\mathbb{X}_{\bar{y}=\bar{b}}|$ for all $\bar{a}\bar{b}\in A^{|\bar{x}\bar{y}|}$ iff $\mathbb{X}\models\sigma$.
\end{proof}
The proofs of Theorems \ref{INDax} and \ref{UFDUINDIAthr} also entail the existence of Armstrong relations.
\begin{corollary}
    The class IND enjoys Armstrong relations over any totally zero-min ordered self absorptive semiring.
\end{corollary}
\begin{corollary}
    The class UFD+UIND+IA enjoys Armstrong relations over any totally zero-min ordered multiplicatively cancellative and commutative idempotent semiring.
\end{corollary}
Let $\mathcal{A}_{FD+UMI}:=\mathcal{A}_{FD}\cup\mathcal{A}_{UMI}\cup\mathcal{A}_{CYCLE-UMI}$, $\mathcal{A}_{FD+UMI+SCI}:=\mathcal{A}_{FD}\cup\mathcal{A}_{UMI}\cup\mathcal{A}_{SCI}\cup\mathcal{A}_{SCI\& FD}\cup\mathcal{A}_{CYCLE-UMI}$ and $\mathcal{A}_{UFD+UMI+IA}:=\mathcal{A}_{UFD}\cup\mathcal{A}_{UMI}\cup\mathcal{A}_{IA}\cup\mathcal{A}_{FD \& IA}\cup\mathcal{A}_{CYCLE-UMI}$. From the proofs of Theorems \ref{FDUINDax} and \ref{UFDUINDIAthr}, we obtain the corresponding results for FDs+UMIs+SCIs and UFDs+UMIs+IAs by using the lemma below:
since in the axiomatizations $\mathcal{A}_{FD+UMI+SCI}$ and $\mathcal{A}_{UFD+UMI+IA}$,
new UMIs can only be deduced by applying rules $\mathcal{A}_{FD+UMI}$,
the proof of Theorem 4.1 from \cite{hirvonen24} can be easily extended to sets $\Sigma\cup\{\tau\}$ of FDs, UMIs and SCIs (or FDs, UMIs, and IAs) to obtain the following.
\begin{lemma}\label{simullem}
    Let $\Sigma\cup\{\tau\}$ be a set of FDs, UMIs and SCIs (FDs, UMIs, and IAs). Then $\Sigma\vdash\tau$ iff $\Sigma^*\vdash\tau^*$, where $\Sigma^*:=\{\sigma\in\Sigma :\sigma \text{ an FD or an SCI (IA) }\}\cup\{x\leq y : x\approx y\in\Sigma \text{ or } y\approx x\in\Sigma\}$ and 
    \[
    \tau^*:=\begin{cases}
        \tau, \text{ if } \tau \text{   is an FD or an SCI (IA)},\\
        x\leq y \ \wedge \ y\leq x, \text{ if } \tau=x\approx y.
    \end{cases}
    \]
In the above, the notation $\Sigma^*\vdash  x\leq y \ \wedge \ y\leq x$ means that both $\Sigma^*\vdash  x\leq y$ and $\Sigma^*\vdash y\leq x$ hold.
\end{lemma}
\begin{proof}
We show the case of UFDs, UMIs, and SCIs; the proof of the claim for UFDs, UMIs, and IAs is analogous. We prove the claim by induction on the length $n$ of the deduction. Note that by \cite{hirvonen24}, we have that if $\Sigma^*\vdash x\leq y$, then also $\Sigma^*\vdash y\leq x$ and the length of the shortest deduction of $y\leq x$ is at most that of $x\leq y$.

 First, suppose that $\Sigma\vdash\tau$. 
\begin{itemize}
    \item In the case of $n=0$, we must have $\tau\in\Sigma$, so clearly $\tau^*\in\Sigma^*$ by the construction of $\Sigma^*$ and $\tau^*$.
    \item The case $n=m+1$. If the last axiom that was used in the deduction of $\Sigma\vdash\tau$ is in $\mathcal{A}_{FD}$, $\mathcal{A}_{SCI}$, or $\mathcal{A}_{SCI\&FD}$, we obtain the deduction for $\Sigma^*\vdash\tau^*$ by using the induction hypothesis and  the fact that the axioms in $\mathcal{A}_{FD}$, $\mathcal{A}_{SCI}$, and $\mathcal{A}_{SCI\&FD}$ are also included in $\mathcal{A}_{FD+UIND+SCI}$. If the last axiom that was used in the deduction of $\Sigma\vdash\tau$ is in $\mathcal{A}_{UMI}$ or $\mathcal{A}_{CYCLE-UMI}$, the case is clear by the induction hypothesis and the corresponding axiom in $\mathcal{A}_{UIND}$ or $\mathcal{A}_{CYCLE-UIND}$. Note that in the case of the symmetry axiom UMI2, there is no corresponding axiom for UINDs, but we obtain the claim from the fact that if $\Sigma^*\vdash x\leq y$, then also $\Sigma^*\vdash y\leq x$.
\end{itemize}
Suppose then that $\Sigma^*\vdash\tau^*$.
\begin{itemize}
    \item The case $n=0$ is again clear by the construction of $\Sigma^*$ and $\tau^*$.
    \item The case $n=m+1$. If the last axiom that was used in the deduction of $\Sigma^*\vdash\tau^*$ is in $\mathcal{A}_{FD}$, $\mathcal{A}_{SCI}$, or $\mathcal{A}_{SCI\&FD}$, we again obtain the deduction for $\Sigma\vdash\tau$ by using the induction hypothesis and the fact that the axioms in $\mathcal{A}_{FD}$, $\mathcal{A}_{SCI}$, and $\mathcal{A}_{SCI\&FD}$ are also included in $\mathcal{A}_{FD+UMI+SCI}$. If the last axiom that was used in the deduction of $\Sigma^*\vdash\tau^*$ is in $\mathcal{A}_{UIND}$ or $\mathcal{A}_{CYCLE-UIND}$, we obtain the deduction of $\Sigma\vdash\tau$ by using the induction hypothesis,  the corresponding axiom in $\mathcal{A}_{UMI}$ or $\mathcal{A}_{CYCLE-UMI}$, and the fact that if $\Sigma^*\vdash x\leq y$, then also $\Sigma^*\vdash y\leq x$ and the length of the shortest deduction of $y\leq x$ is at most that of $x\leq y$.
\end{itemize}
\end{proof}
Now, the lemma gives us the following corollary.
\begin{corollary}\label{simulcor}
        The axiomatizations $\mathcal{A}_{FD+UMI+SCI}$ and $\mathcal{A}_{UFD+UMI+IA}$ are sound and complete for the implication problems for FDs+UMIs+SCIs and UFDs+UMIs+IAs, respectively, over any 
        multiplicatively cancellative and commutative idempotent semiring.
\end{corollary}
\begin{proof}
   Let $\Sigma\cup\{\tau\}$ be a set of FDs, UMIs, and SCIs such that $\Sigma\not\vdash\tau$. Define $\Sigma^*$ and $\tau^*$ as in Lemma \ref{simullem}, and note that then $\Sigma^*\not\vdash\tau^*$. We show that $\Sigma\not\models\tau$. 
Suppose first that $\tau$ is an FD or an SCI. By the proof of Theorem \ref{FDUINDax}, from $\Sigma^*\not\vdash\tau^*$, it follows that there is a $K$-team $\X$ such that $\X\models\Sigma^*$ and $ \X\not\models\tau^*$. Clearly, we have $\X\models\Sigma$ and $\X\not\models\tau$.
   If $\tau=x\approx y$, then by the construction of $\Sigma^*$, $\Sigma^*\not\vdash\tau^*$ means that $\Sigma^*\not\vdash x\leq y$ (and $\Sigma^*\not\vdash y \leq x$). Then again by the proof of Theorem \ref{FDUINDax}, it follows that there is a $K$-team $\X$ such that $\X\models\Sigma^*$ and $ \X\not\models x\leq y$. Clearly, then we have $\X\models\Sigma$ and $\X\not\models\tau$.
   The proof of the claim for UFDs, UMIs and IAs is analogous.
\end{proof}

\section{Non-existence of finite axiomatizations}\label{nonexist}

The infiniteness of the axiomatizations in the previous sections stems from having an infinite rule schema that contains the $\ell$-cycle rules for all odd $\ell$. In this section, we show that there is no way to get rid of this cyclicity: the implication problems of the previous sections whose axiomatization contains a version cycle rule schema cannot be finitely axiomatized. 

An inference rule $\Delta\vdash \sigma_1,\dots,\sigma_n$ is called $\ell$-ary if $|\Delta|=\ell$. A set of inference rules $\mathcal{A}$ is called $\ell$-ary if it only contains rules that are at most $\ell$-ary.
It is known that in the case of probability distributions, the cycle rules for UFDs and UMDEs cannot be condensed into shorter rules to obtain a $\ell$-ary axiomatization, as seen in the following lemma.
\begin{lemma}[\cite{hirvonen24}]\label{noaxlemma}
    Let $\Sigma:=\{\dep(x_{i},x_{S_\ell(i)}):0\leq i\leq \ell-1, i\text{ even}\}\cup\{x_{i}\approx^*x_{S_\ell(i)}:1\leq i\leq \ell, i\text{ odd}\}$ and $\tau:=\dep(x_1,x_0)$. Then for any $\ell$-ary set $\mathcal{A}_\ell$ of sound inference rules for FD+UMI+UMDE implication over $\mathbb{R}_{\geq0}$, we have $\Sigma\not\vdash_{\mathcal{A}_\ell}\tau$. 
\end{lemma}
From the lemma, it easily follows that the analogous result holds for any additively cancellative semiring.
\begin{corollary}\label{noaxcor}
    Let $K$ be an additively cancellative semiring, $\Sigma:=\{\dep(x_{i},x_{S_\ell(i)}):0\leq i\leq \ell-1, i\text{ even}\}\cup\{x_{i}\approx^*x_{S_\ell(i)}:1\leq i\leq \ell, i\text{ odd}\}$ and $\tau:=\dep(x_1,x_0)$. Then for any $\ell$-ary set $\mathcal{A}_\ell$ of sound inference rules for FD+UMI+UMDE implication over $K$, we have $\Sigma\not\vdash_{\mathcal{A}_\ell}\tau$. 
\end{corollary}
\begin{proof}
    Let $\mathcal{A}_\ell$ be a $\ell$-ary set of  sound inference rules for FD+UMI+UMDE implication over an additively cancellative semiring $K$, and $\Delta\vdash \sigma_1,\dots,\sigma_n$ be an inference rule in $\mathcal{A}_\ell$. Since the axiomatization $\mathcal{A}_{FD+UMI+UMDE}$ is sound and complete over any additively cancellative semiring, we have $\Delta\models_K\sigma_1,\dots,\sigma_n$ iff $\Delta\vdash_{\mathcal{A}_{FD+UMI+UMDE}}\sigma_1,\dots,\sigma_n$ iff $\Delta\models_{\mathbb{R}_{\geq0}}\sigma_1,\dots,\sigma_n$. As $\Delta\vdash \sigma_1,\dots,\sigma_n$ is sound over $K$, this means that $\Delta\vdash \sigma_1,\dots,\sigma_n$ is a sound inference rule 
    over $\mathbb{R}_{\geq0}$, and hence $\mathcal{A}_\ell$ is sound 
    over $\mathbb{R}_{\geq0}$. Then by Lemma \ref{noaxlemma}, we have $\Sigma\not\vdash_{\mathcal{A}_\ell}\tau$. 
\end{proof}
The following nonexistence result is now immediate.
\begin{theorem}\label{FDUMIUMDEnofinite}
    There is no finite complete axiomatization for the implication problem of FD+UMI+UMDE over any additively cancellative semiring. 
\end{theorem}
\begin{proof}
    Let $K$ be an additively cancellative semiring, and $\Sigma$ and $\tau$ be as in Corollary \ref{noaxcor}. Suppose for a contradiction that $\mathcal{A}$ is a   finite complete axiomatization for the implication problem of FD+UMI+UMDE over $K$. Then there is some $\ell$ such that $\mathcal{A}$ is $\ell$-ary, and by Corollary \ref{noaxcor}, we have $\Sigma\not\vdash_{\mathcal{A}}\tau$. Since $\Sigma\vdash_{\mathcal{A}_{FD+UMI+UMDE}}\tau$, we have $\Sigma\models_K\tau$, so $\mathcal{A}$ is not complete.
\end{proof}
The proof of an analogous nonexistence result for FDs+UINDs in the case of $K=\mathbb{B}$ for multirelational databases in \cite{casanova82} can be easily adapted to the unirelational semiring setting similarly to \cite{hirvonen24}.
\begin{lemma}\label{noaxindlemma}
    Let $K$ be a totally zero-min ordered self absorptive semiring, $\Sigma:=\{\dep(x_{i},x_{S_\ell(i)}):0\leq i\leq \ell-1, i\text{ even}\}\cup\{x_{S_\ell(i)}\leq x_{i}, 1\leq i\leq \ell, i\text{ odd}\}$, and $\tau:=x_0\leq x_\ell$. Then for any $\ell$-ary set $\mathcal{A}_\ell$ of sound inference rules for FD+UIND implication over $K$, we have $\Sigma\not\vdash_{\mathcal{A}_\ell}\tau$. 
\end{lemma}
\begin{proof}
    Let $K$, $\Sigma$, and $\tau$ be as in the statement on the lemma. Let $\mathcal{A}$ be the set of all sound, at most $\ell$-ary inference rules for FD+UIND implication. Define $\Gamma:=\Sigma\cup\{x_i\leq x_i : 0\leq i\leq \ell\}\cup\Sigma^*$, where $\Sigma^*$ is defined inductively as follows:
    \begin{itemize}
        \item if $\Var(\bar{z})\subseteq\Var(\bar{y})\subseteq\{x_0,\dots,x_\ell\}$, then $\dep(\bar{y},\bar{z})\in\Sigma^*$
        \item if $\dep(\bar{y},\bar{z})\in\Sigma^*$ and $i\in\{0,2,\dots,\ell-1\}$, then $\dep(\bar{y}x_i,\bar{z}x_{S_\ell(i)})\in\Sigma^*$,
        \item if $\dep(\bar{y},\bar{z})\in\Sigma^*$ and $\Var(\bar{y}')=\Var(\bar{y}),\Var(\bar{z}')=\Var(\bar{z})$, then $\dep(\bar{y}',\bar{z}')\in\Sigma^*$.
    \end{itemize}
    It is easy to check that $\Gamma\subseteq\Cl_{\mathcal{A}}(\Sigma)$. We show that $\Cl_{\mathcal{A}}(\Sigma)\subseteq\Gamma$ by induction on the length of $\mathcal{A}$-deduction for $\sigma\in\Cl_{\mathcal{A}}(\Sigma)$.
    \begin{itemize}
        \item The case $n=0$ is clear, since it implies that $\sigma\in\Sigma$, and we have $\Sigma\subseteq\Gamma$.
        \item Let $n=m+1$, and suppose that $\Delta:=\{\delta_1,\dots,\delta_\ell\}$ is such that at the last step of deduction,  $\sigma$ is obtained from $\Delta$. Suppose for a contradiction that $\sigma\not\in\Gamma$. We show that then $\Delta\not\models\sigma$. Let $\delta\in\Sigma\setminus\Delta$. By the cyclicity of the cycle rules, it suffices to consider the cases $\delta:=\dep(x_0,x_1)$ and $\delta:=x_\ell\leq x_0$. We only show the case $\delta:=\dep(x_0,x_1)$, as the case for $\delta:=x_\ell\leq x_0$ is similar. By the induction hypothesis $\Delta\subseteq\Gamma\setminus\{\delta\}$, so $\Gamma\setminus\{\delta\}\models\Delta$. Let $\Sigma^{0}:=\{\dep(x_0x_i,x_1
        ):i\in\{2,4,\dots,\ell-1\}\}$. It is straightforward to check that $\Sigma\setminus\{\delta\}\cup \Sigma^{0}\models\Gamma\setminus\{\delta\}$, so it suffices to show that $\Sigma\setminus\{\delta\}\cup \Sigma^{0}\not\models\sigma$. This can be done by constructing a team $X$ 
        as in Definition \ref{uindfd_def} of Appendix \ref{definitions}, and making it a $K$-team $\mathbb{X}$ by setting $\mathbb{X}(s)=k$, where $k$ is an idempotent element of $K$. Since $\Delta\not\models\sigma$ contradicts the soundness of $\mathcal{A}$, we must have $\sigma\in\Gamma$.
    \end{itemize}
    We then have $\Gamma=\Cl_{\mathcal{A}}(\Sigma)$, and since $\tau\not\in\Gamma$, the claim follows.
\end{proof}
Now Lemma \ref{noaxindlemma} gives us the following theorem by an argument analogous to the proof of Theorem \ref{FDUMIUMDEnofinite}.
\begin{theorem}
    There is no finite complete axiomatization for the implication problem of FD+UIND over any totally zero-min ordered self absorptive semiring. 
\end{theorem}
The result extends to the cases FD+UIND+SCI and UFD+UIND+IA over any totally zero-min ordered multiplicatively cancellative idempotent semiring $K$.
\begin{theorem}\label{nofiniteax2}
    There is no finite complete axiomatization for the implication problem of FD+UIND+SCI or UFD+UIND\ +IA over any totally zero-min ordered multiplicatively cancellative idempotent semiring. 
\end{theorem}
\begin{proof} 
The result follows from the nonexistence of $\ell$-ary axiomatizations for FD+UIND+SCI or UFD+UIND+IA, so the proof is similar to the proof of Lemma \ref{noaxindlemma}. We describe the idea for the case FD+UIND+SCI, the case FD+UIND+IA is analogous. The main changes from the case FD+UIND to the case FD+UIND+SCI are that $\mathcal{A}$ is now the set of all sound, at most $\ell$-ary inference rules for FD+UIND+SCI implication and the construction of $\Sigma^*$ is extended with the following two clauses:
    \begin{itemize}
        \item if $\dep(\bar{y},\bar{z})\in\Sigma^*$, then $\bar{y}\twoheadrightarrow\bar{z}\in\Sigma^*$, 
        \item if $\bar{y}\twoheadrightarrow\bar{z}\in\Sigma^*$, $\Var(\bar{z})\cap\Var(\bar{u})\subseteq\Var(\bar{y})$, and $\Var(\bar{y}\bar{z}\bar{u})=\{x_0,\dots,x_\ell\}$, then $\bar{y}\twoheadrightarrow\bar{u}\in\Sigma^*$.
    \end{itemize}
  It is easy to see that the SCIs in $\Gamma\setminus\{\delta\}$ do not affect the construction of $X$ because they either follow from the corresponding FDs or are trivially satisfied by symmetry. If $\sigma:=\bar{y}\twoheadrightarrow\bar{z}$, then it suffices to construct the team $X$ for the case $\sigma:=\dep(\bar{y},\bar{z})$.\footnote{Since $\sigma:=\dep(\bar{y},\bar{z})\notin\Gamma$, we also have $\dep(\bar{y},x_0\dots x_\ell\setminus\bar{y}\bar{z})\notin\Gamma$. It is easy to check that from this it follows that the two tuples added in the step (0) in the construction of $X$ in Definition \ref{uindfd_def} ensure that $X\not\models\bar{y}\twoheadrightarrow\bar{z}$.}
\end{proof}
\begin{theorem}
    There is no finite complete axiomatization for the implication problem of FD+UMI+SCI or UFD+UMI\ +IA over any 
    multiplicatively cancellative idempotent semiring. 
\end{theorem}
\begin{proof}
    Since the proof is similar to the other proofs in this section, we only give a sketch of the proof idea. Let $\Sigma:=\{\dep(x_{i},x_{S_\ell(i)}):0\leq i\leq \ell-1, i\text{ even}\}\cup\{x_{i}\approx x_{S_\ell(i)}:1\leq i\leq \ell, i\text{ odd}\}$ and $\tau:=\dep(x_1,x_0)$. Then it is easy to see that for any $\ell$-ary set $\mathcal{A}_\ell$ of sound inference rules for FD+UMI implication over $\mathbb{R}_{\geq0}$, we have $\Sigma\not\vdash_{\mathcal{A}_\ell}\tau$. Now, the proof is analogous to the proof of Theorem \ref{nofiniteax2}.
\end{proof}
\section{Complexity results}

In the literature, there exist polynomial time algorithms for the implication problems for FDs+UINDs+SCIs \cite{CosmadakisKV90}, UFDs+UINDs+IAs \cite{hannula2021interactionfunctionalinclusiondependencies}, FDs+UMIs+UMDEs \cite{hirvonen24}, and IAs+UMIs+UMDEs \cite{hirvonen24conf} over either relational databases or probability distributions. Since we showed that the same (or analogous) axioms are sound and complete more generally for certain classes of semirings, we obtain the following corollary.
\begin{corollary}
The implication problem is in\footnote{Note that by Remark \ref{generalization}, the condition of additive cancellativity in (iii) and (iv) could also be replaced with the condition that the semiring contains an additive submonoid that is isomorphic to $(\mathbb{N},+,0)$ without an order.} 
\begin{itemize}
        \item[(i)] linear time for UMDIs over any weakly absorptive semiring,
        \item[(ii)]  polynomial time for FDs+UINDs+SCIs, FDs+UMI+SCIs, UFDs+UINDs+IAs, and UFDs+ UMIs+IAs, over any totally zero-min ordered multiplicatively cancellative and commutative idempotent semiring,
   \item[(iii)]  polynomial time for FDs+UMIs+UMDEs, over any additively cancellative semiring,
      \item[(iv)]  polynomial time for UMIs+UMDEs+IAs, over any additively and multiplicatively cancellative semiring.
\end{itemize}    
\end{corollary}
\begin{proof}
    Excluding the classes UMDIs, FDs+UMI+SCIs and UFDs+UMIs+IAs, for each class of atoms $\mathcal{C}$ in the corollary, we have $\Sigma\models_K\tau\iff\Sigma\vdash_{\mathcal{A}_{\mathcal{C}}}\tau\iff\Sigma\models_{K'}\tau$, where $K$ is any semiring from the corresponding class of semirings for $\mathcal{C}$ in the corollary, $\mathcal{A}_{\mathcal{C}}$ is the axiomatization for the class $\mathcal{C}$, and $K'$ is either $\mathbb{B}$ or $\mathbb{R}_{>0}$ depending on whether the existing algorithm is for relational databases or probability distributions. Hence, to decide whether $\Sigma\models\tau$ over the relevant class of semirings, it suffices to use the existing algorithm for $\mathcal{C}$ to check whether $\Sigma\models_{K'}\tau$.

    For the class UMDIs, note that since the axioms $\mathcal{A}_{UMDI}$ and $\mathcal{A}_{UIND}$ are analogous, for any weakly absorptive semiring $K$, we have $\Sigma\models_K\tau\iff\Sigma\vdash_{\mathcal{A}_{UMDI}}\iff\Sigma'\vdash_{\mathcal{A}_{UIND}}\tau'\iff\Sigma'\models_{\mathbb{B}}\tau'$, where $\Sigma'$ and $\tau'$ are obtained from $\Sigma$ and $\tau$ by replacing each UMDI $x\subseteq^*y$ by the corresponding UIND $x\leq y$. The problem of checking $\Sigma'\models_{\mathbb{B}}\tau'$ is essentially the same as directed graph reachability \cite{CosmadakisKV90}, which is in linear time.
    
    For the classes $\mathcal{C}$ of FDs+UMI+SCIs and UFDs+UMIs+IAs, we can use the fact that $\Sigma\models_K\tau\iff\Sigma\vdash_{\mathcal{A}_{\mathcal{C}}}\tau\iff\Sigma^*\vdash_{\mathcal{A}_{\mathcal{C}^*}}\tau^*\iff\Sigma^*\models_{\mathbb{B}}\tau^*$, where $K$ and $\mathcal{A}_{\mathcal{C}}$ are as above, and $\mathcal{C}^*$ is the corresponding class that contains UINDs instead of UMIs, and $\Sigma^*$ and $\tau^*$ are as in Lemma \ref{simullem}. Since the construction of $\Sigma^*\cup\{\tau^*\}$ from $\Sigma\cup\{\tau\}$ can be done in polynomial time, the implication problems can be solved in polynomial time by constructing $\Sigma^*\cup\{\tau^*\}$ and then checking whether $\Sigma^*\models_{\mathbb{B}}\tau^*$.
\end{proof}

\section{Conclusion and Open Questions}

In this paper, we have studied implication problems for various classes of atoms in semiring team semantics, and have shown that many known axiomatizations generalize to suitable classes of semirings. The properties of the semiring impact which atoms are sensible to consider; in particular, additive cancellativity seems to be relevant for a symmetry property of certain atoms. Below, we list some questions that are left open.
\begin{itemize}
    \item[(i)] Can the completeness results for some classes of non-additively cancellative semirings be improved to require less specific assumptions?
    \item[(ii)] Are the axioms MDI1--MDI3 (MDE1--MDE4) complete for marginal distribution inclusion (equivalence) atoms over some classes of semirings? 
    \item[(iii)] Can some results of Section \ref{withoutaddcan} be extended to obtain complete axiomatizations for classes that contain UMDIs? 
    \item[(iv)] Can a result similar to Corollary \ref{simulcor} be obtained so that the axiomatizations could be extended to classes that contain UMDEs, perhaps by restricting to implication over additively cancellative semirings?
\end{itemize}
Most of this paper's proofs for non-additively cancellative semirings rely on the assumption that the 
semiring under consideration contains a self absorptive nonzero element that can be used as the weight of each assignment in the constructed teams. This allows us to easily generalize the relational (Boolean) case to other semirings. Without a self absorptive nonzero element, balancing the satisfaction (and the nonsatisfaction) of UINDs (or UMDIs) with FDs or IAs in the constructed teams seems to become more difficult\footnote{In particular, if one tries to use a construction similar to that of Definition \ref{uindfd_def} for proving completeness in the case of FD+UIND without a self absorptive element, one needs to overcome the problem that the number of zeroes in different columns increases with the scc-number of the component while the necessary UINDs still need to be satisfied. One way to try to fix this would be to use nonzero elements $k,h\in K$ such $k+h=k$, and give the weight $k$ to the first (all-zero) row, and the weight $h$ to each of the other rows, so that the weights of the subsequent zeroes are absorbed by the weight of the first row. This still seems to cause problems for the constancy atoms in step (ivb) and for the rows with $a$'s in step (0) of Definition \ref{uindfd_def}.}. In the case of UMI and UMDE, the existence of a submonoid similar to $\mathbb{N}$ seems to help (see Remark \ref{generalization} regarding the cases FD+UMI+UMDE and UMI+UMDE+IA) because the symmetry property of the atoms ensures that distributing the weights uniformly works. Note that the interaction of UINDs and UMDIs also differs from the UMI+UMDE case: $x\approx y$ implies $x\approx^* y$, but  $x\leq y$ does \emph{not} imply $x\subseteq^* y$.

In practice, the assumption of a self absorptive nonzero element might not be overly restrictive: the examples of non-additively cancellative semirings given in the preliminaries section are all idempotent, so the results with the assumptions used in this paper nonetheless cover some well-known semirings. On the other hand, the requirement of zero-min order for weighted marginal inclusion dependencies might not be desirable for some applications that, e.g., use tropical semirings for optimization, because we might actually want that the nonexistence of some value in the support prohibits the satisfaction of a weighted marginal inclusion dependency. 
Note that the results cannot be extended to obtain complete axiomatizations for any classes of atoms containing FDs+INDs or FDs+IAs over all positive semirings. In the first case, the reason is that the implication problem for FDs+INDs is undecidable\footnote{If an implication problem is undecidable, there cannot be a variable bounded axiomatization for it. (We say that an axiomatization $\mathcal{A}$ is variable bounded if every deduction $\Sigma\vdash_{\mathcal{A}}\tau$ can be witnessed by a sequence $(\tau_1,\dots,\tau_n)$ such that $Var(\{\tau_1,\dots,\tau_n\})\subseteq\Var(\Sigma\cup\{\tau\})$.) If there was a variable bounded axiomatization, then given a finite set of atoms $\Sigma\cup\{\tau\}$, such that $\Var(\Sigma\cup\{\tau\})=D$, there are only finitely many possible deductions for $\Sigma\vdash\tau$ using the variables in $D$. An algorithm checking these deductions would determine whether $\Sigma\models\tau$, making the implication problem decidable.} over the Boolean semiring \cite{chandra85}. In the latter case, this follows from the fact that the implication problem for FDs+IAs is an undecidable fragment of the implication problem for (non-disjoint) conditional independence \cite{LI23}.










\appendix

\section{Additional Details on Semirings}\label{exapp}

This section contains additional 
examples and details that where omitted from the main text.

\begin{table}
\centering
        \begin{tabular}{c c c | c }
        \hline
        $ \ x \ $ & $ \ y \ $ & $ \ z \ $ &  $ \ \mathbb{X} \ $  \\
        \hline
        0 & 1 & 1 & 1  \\
        0 & 0 & 2 & 1  \\
        1 & 0 & 0 & 1  \\
        0 & 0 & 0 & -1  \\
        \hline 
        \ ~ \
    \end{tabular}
\caption{The $\mathbb{Z}$-team $\mathbb{X}$ of Example \ref{nonpos2} demonstrating the failure of transitivity for the ``local-by-definition'' variant of FDs.}
\label{ex_table2}
\end{table}

\begin{example}\label{nonpos2}
        Consider the semiring of integral numbers $\mathbb{Z}=(\mathbb{Z},+,\times,0,1)$ that is not positive. Define the $\mathbb{Z}$-team $\mathbb{X}$ of $A=\{0,1,2\}$ over $D=\{x,y,z\}$ as depicted in Table \ref{ex_table2}. Then $\X\models\dep(x,y)^{*}$ and $\X\models\dep(y,z)^{*}$, but $\X\not\models\dep(x,z)^{*}$, demonstrating the failure of transitivity for the ``local-by-definition'' variant of FDs.
\end{example}

\begin{proposition}\label{polynom}
  Define the order of $\mathbb{N}[x]=(\mathbb{N}[x],+,\times,0,1)$ such that $P(x)\leq Q(x)$ iff $P(x)=Q(x)$ or $P(x)$ and $ Q(x)$ can be written as $P(x)=k_0x^0+k_1x^1+\dots+k_nx^n+\sum_{i=n+1}^{\infty} k_{i}x^i$
 and $Q(x)=h_0x^0+h_1x^1+\dots+h_nx^n+\sum_{i=n+1}^{\infty} k_{i}x^i$ where $k_n< h_n$ in the usual order of $\mathbb{N}$. The order $\leq$ is a zero-min order, but not the natural order.  
\end{proposition}
\begin{proof}
    It is straightforward to check that $\mathbb{N}[x]$ is indeed a semiring ordered by $\leq$. We show that $\leq$ is zero-min. Let $P(x)\leq Q(x)$, where $Q(x)=0$. If $P(x)=Q(x)=0$, we are done, so suppose for a contradiction that $P(x)$ and $ Q(x)$ can be written as  $P(x)=k_0x^0+k_1x^1+\dots+k_nx^n+\sum_{i=n+1}^{\infty} k_{i}x^i$ and $Q(x)=h_0x^0+h_1x^1+\dots+h_nx^n+\sum_{i=n+1}^{\infty} k_{i}x^i$ where $k_n< h_n$. Since $Q(x)=0$, we have $h_n=0$, and hence also $k_n<0$, which is not possible.

    We now show that $\leq$ is not the natural order. Consider the case where $P(x)=1$ and $Q(x)=x$. Now $P(x)=1x^0+0x^1+\sum_{i=2}^{\infty} 0x^i$ and $Q(x)=0x^0+1x^1+\sum_{i=2}^{\infty} 0x^i$ where $k_2=0<1=h_2$, so $P(x)\leq Q(x)$. However, there is no $T(x)\in\mathbb{N}[x]$ such that $P(x)+T(x)=Q(x)$.
\end{proof}
\noindent\textbf{Lemma \ref{addcan_lemma}.}
 If $K$ is additively cancellative, then $k<k'$ and $h\leq h'$ imply $k+h<k'+h'$. 
\begin{proof}
  Since $K$ is additively cancellative, $k\neq k'$ implies $k+h\neq k'+h$, and therefore $k<k'$ implies $k+h<k'+k$. If $k<k'$ and $h\leq h'$, then $k+h<k'+h$ and $k'+h\leq k'+h'$. Then also $k+h\leq k'+h'$. Clearly, $k+h\neq k'+h'$, because otherwise $k'+h\leq k+h$, contradicting $k+h<k'+h$. This means that $k+h<k'+h'$. 
\end{proof}
\begin{proposition}\label{doubletwo}
    Let $\mathbb{N}'=(\mathbb{N}\cup\{2'\},\oplus,\otimes,0,1)$ where $\oplus$ and $\otimes$ are the usual addition and multiplication on the natural numbers extended to $\mathbb{N}\cup\{2'\}$ by letting $2'\oplus 0=0\oplus 2'=2'$, $2'\oplus  2'=4$, and   $2'\oplus k=k\oplus 2'=2+k$ for all $k\in\mathbb{N}\setminus\{0\}$, and $2'\otimes 0=0\otimes 2'=0$, $2'\otimes 1=1\otimes 2'=2'$, $2'\otimes  2'=4$, and $2'\otimes h=h\otimes 2'=2\times h$ for all $h\in\mathbb{N}\setminus\{0,1\}$. Then $\mathbb{N}'$ is weakly additively cancellative, but not additively cancellative. Moreover, the order $\leq$ defined by $0<1<2'<2<3<4<\dots$ is zero-min, but not the natural order.
\end{proposition}
\begin{proof}
    It is straightforward to check that $\mathbb{N}'$ is indeed a semiring ordered by $\leq$. Clearly, $k\oplus h=k$ implies $h=0$, so $\mathbb{N}'$ is weakly additively cancellative. On the other hand $2\oplus 1=3=2'\oplus 1$, but $2\neq 2'$, so $\mathbb{N}'$ is not additively cancellative. The order $\leq$ is zero-min by definition. It is not the natural order because $2'\leq 2$, but there is no $k\in\mathbb{N}\cup\{2'\}$ such that $2'+k=2$.
\end{proof}
The semiring $\mathbb{N}'$ with the order defined as in Proposition \ref{doubletwo} might seem artificial. However, one could interpret it as modelling support\footnote{Note that here ``support'' refers to the meaning of the word in the usual sense (endorsement, approval), i.e., it does not mean the support of the team, $\Supp(\mathbb{X})$.} (votes, evidence etc.) for different options in the case where we want to decrease the weight of a possible ``double support'' that comes from a single source. Then the annotation $2'$ for an option refers to having slightly less support than the annotation 2, and any additional support is deemed enough for this difference to disappear, refering to the fact that, $2'\oplus k=2+k$ for any $k\in\mathbb{N}\setminus\{0\}$. For example, in the case of voting, $\mathbb{X}(s_0)=2'$ could refer to two votes for an option $s_0$ from the same person (that has two votes, perhaps due to a dual role), whereas $\mathbb{X}(s_1)=2$ could refer to two votes for an option $s_1$ from two different people. The idea is that we want to decrease the power of double votes from people with two votes in the case that the option they voted for does not receive additional support. If $\mathbb{X}$ is now restricted to some smaller set of variables $V\subsetneq D$, 
i.e., we consider $\mathbb{X}\restriction V$, then the effect of decreased support for option $s_0$ disappears in $\mathbb{X}\restriction V$ if there is $s\in X\setminus\{s_0\}$ such that $\mathbb{X}(s)\in\mathbb{N}'\setminus\{0\}$ and $s\restriction V=s_0\restriction V$. 
\begin{remark}
    The team $\mathbb{X}$ of Table \ref{ex_table} in the case $K=\mathbb{N}'$, $k=k'=2'$, and $h=2$ demonstrates the failure of symmetry for the atom $x\leq y$ for a semiring that is weakly additively cancellative, but not additively cancellative: we have $|\mathbb{X}_{x=0}|=4=|\mathbb{X}_{y=0}|$ and $|\mathbb{X}_{x=1}|=2'<2=|\mathbb{X}_{y=1}|$, so $\mathbb{X}\models x\leq y$, but $\mathbb{X}\not\models y\leq x$.
\end{remark}  
\begin{proposition}\label{naturordprop}
    If $K=(K,+,\times,0,1)$ is a positive weakly additively cancellative semiring that is totally ordered by the natural order, then $K$ is also additively cancellative. 
\end{proposition}
\begin{proof}
    Suppose for a contradiction that there are $h,k,k'\in K$ such that $h+k=h+k'$, but $k\neq k'$. Because $K$ is totally ordered, we may assume that $k\leq k'$. Then there is $h'\in K$ such that $k+h'=k'$. Then $h+k=h+k'$ can be written as $h+k=h+k+h'$. By weak additive cancellativity, this implies that $h'=0$, so $k'=k+h'=k$, which is a contradiction.
\end{proof}
\section{Additional Definitions and Proof Details for Lemmas \ref{noaxindlemma} and \ref{FDUMIUMDE_lemma}}\label{definitions}

The following definitions and lemmas for UINDs+FDs are used for the construction of the team $X$ that is needed in Lemma \ref{noaxindlemma}.
\begin{definition}[\cite{CosmadakisKV90}]\label{graph_def2}
Let $\Sigma$ be a set of FDs and UINDs such that $\Var(\Sigma)\subseteq D$. Denote by  $G(\Sigma)$ a multigraph defined as follows
\begin{itemize}
\item[(i)] the vertices of $G(\Sigma)$ are the variables $D$,
\item[(ii)] if $x\leq y\in\Sigma$, then there is an directed black edge from $y$ to $x$,
\item[(iii)] if $\dep(x,y)\in\Sigma$, then there is a directed red edge from $x$ to $y$.
\end{itemize}
\end{definition}
If there are two black (red) directed edges from $x$ to $y$ and from $y$ to $x$, they are viewed as a single black (red) undirected edge between $x$ and $y$.
\begin{lemma}[\cite{CosmadakisKV90}]\label{graph_prop_lemma2}
Let $\Sigma$ be a set of FDs and UINDs  such that $\Var(\Sigma)\subseteq D$ and let $\Delta:=\Cl(\Sigma)$. Then the graph $G(\Delta)$ has the following properties.
\begin{itemize}
\item[(i)] Each vertex has a black and red self-loop.
\item[(ii)] The black and red subgraphs of $G(\Delta)$ are transitively closed. 
\item[(iv)] The subgraphs induced by the strongly connected components of $G(\Delta)$   are undirected. Each such component contains a black and red undirected subgraph. In the subgraph of each color, the vertices of the component can be partitioned into a collection of disjoint cliques.
\item[(v)] If $\dep(\bar{x},y)\in\Delta$ and the vertices $\bar{x}$ have a common ancestor $z$ in the red subgraph of $G(\Delta)$, then there is a red edge from $z$ to $y$.
\end{itemize}
\end{lemma}
Now we assign a unique number, called the \emph{scc-number}, to each strongly connected component of $G(\Delta)$ such that the scc-number of a descendant component is always greater than the number of its ancestor component. Denote by $\text{scc}(x)$ the scc-number of the strongly connected component that $x$ belongs to.

\begin{definition}[\cite{CosmadakisKV90}]\label{uindfd_def}
Let $\Sigma$ and $G(\Delta)$ be as in Lemma \ref{graph_prop_lemma2} and assign an scc-numbering to  $G(\Delta)$. Let $\tau$ be an FD or UIND such that $\Var(\tau)\subseteq D$ and $\tau\notin\Delta$. Define a team $X$ over the variables $D$ as follows.
\begin{itemize}
\item[(0)] If $\tau$ is a UFD or a UIND, proceed directly to step (i). If $\tau$ is a nonunary FD, do the following before proceeding to step (i). Suppose that $\tau:=\dep(\bar{x},\bar{y})$. Add a tuple of all $a$'s, i.e., an assignment $s$ such that $s(x)=a$ for all $x\in D$. Then add a tuple with $a$'s in exactly those positions $z\in D$ that are functionally determined by $\bar{x}$, i.e., those $z$ for which $\dep(\bar{x},z)\in\Delta$. Leave all the other positions in the tuple empty for now.
\item[(i)] Add a tuple of all zeroes, i.e., an assignment $s$ such that $s(x)=0$ for all $x\in D$.
\item[(ii)] Process each strongly connected component in turn, starting with the one with the smallest scc-number and then proceeding in the ascending order of the numbers. For each strongly connected component, handle each of its maximal red cliques in turn.
\begin{itemize}
\item[(a)] For each maximal red clique $k$, add a tuple with zeroes in the columns corresponding to the variables in $k$ and to the variables that are in any red clique that is a red descendant of $k$. Leave all the other positions in the tuple empty for now.
\item[(b)] Choose a variable in $k$ and count the number of zeroes and $a$'s in a column corresponding to the chosen variable. It suffices to consider only one variable, because the construction ensures that all the columns corresponding to variables in $k$ have the same number of zeroes and $a$'s. Denote this number by $\text{count}(k)$.
\end{itemize}
After adding one tuple for each maximal red clique, check that the $\text{count}(k)$ is equal for every clique $k$ in the current component and strictly greater than the count of each clique in the previous component. If it is not, repeat some of the tuples added to make it so.
\item[(iii)] The last component is a single red clique consisting of those variables $x$ for which $\dep(x)\in\Delta$, if any. Each variable in this clique functionally depends on all the other variables in the graph, so we do not leave any empty positions in its column. Therefore the columns corresponding to these variables contain only zeroes. If there are no variables $x$ for which $\dep(x)\in\Delta$, we finish processing the last component by adding one tuple with all positions empty.
\item[(iv)] After all strongly connected components have been processed, we start filling the empty positions. For each column, count the number of empty positions. If the column has $d$ empty positions, fill them with numbers $1,\dots, d$ without repetitions.
\item[(iv)] After all strongly connected components have been processed, we start filling the empty positions. For each column, count the number of empty positions.
\begin{itemize}
    \item[(a)] If $\tau$ is a UFD, a nonunary FD, or an UIND $x\leq y$ such that $\text{scc}(x)<\text{scc}(y)$ do as follows: for each column, if the column has $d$ empty positions, fill them with numbers $1,\dots, d$ without repetitions.
    \item[(b)]  Otherwise, $\tau$ is a UIND $x\leq y$ such that $\text{scc}(x)\geq \text{scc}(y)$, and do the following. For each variable $z$, if the column for $z$ has $d$ empty positions and there is no black edge from $z$ to $x$, fill the empty positions with numbers $1,\dots, d$ without repetitions. If there is a black edge from $z$ to $x$, fill the empty positions with numbers $1,\dots, d-1,d_{max}+1$, where $d_{max}$ is the maximum number of empty positions in any column. Finally, if $z$ is such that $\dep(z)\in\Delta$ (i.e.,the column $z$ contains only zeroes) and there is a black edge from $z$ to $x$, then change all the zeroes in the column $z$ into $d_{max}+1$'s.
\end{itemize}
\end{itemize}
\end{definition}
The construction ensures that $X$ satisfies the UINDs and nonunary FDs in $\Delta$, and that $X$ satisfies exactly those unary FDs or constant atoms that are in $\Delta$. 

\begin{lemma}[\cite{CosmadakisKV90}]
    Let $\Sigma\cup\{\tau\}$ be a set of FDs and UINDs such that $\Sigma\not\vdash\tau$. Then $X\models\Sigma$ and $X\not\models\tau$, where the $(\mathbb{B}\text{-})$ team $X$ is constructed as in Definition \ref{uindfd_def}.
\end{lemma}

The following definitions and lemma describe the construction of the team $X$ of Lemma \ref{FDUMIUMDE_lemma}. The multigraph construction is similar to Definition \ref{graph_def2}, but the black edges now correspond to UMIs and hence are undirected, and the graph also contains blue edges that correspond to UMDEs.
\begin{definition}[\cite{hirvonen24}]\label{graph_def}
Let $\Sigma$ be a set of FDs, UMIs, and UMDEs such that $\Var(\Sigma)\subseteq D$. Denote by  $G(\Sigma)$ a multigraph defined as follows
\begin{itemize}
\item[(i)] the vertices of $G(\Sigma)$ are the variables $D$,
\item[(ii)] if $x\approx y\in\Sigma$, then there is an undirected black edge between $x$ and $y$,
\item[(iii)] if $x\approx^* y\in\Sigma$, then there is an undirected blue edge between $x$ and $y$, 
\item[(iv)] if $\dep(x,y)\in\Sigma$, then there is a directed red edge from $x$ to $y$.
\end{itemize}
\end{definition}
As before, if there are two red directed edges from $x$ to $y$ and from $y$ to $x$, they are viewed as a single red undirected edge between $x$ and $y$.
\begin{lemma}[\cite{hirvonen24}]\label{graph_prop_lemma}
Let $\Sigma$ be a set of FDs, UMIs and UMDEs such that $\Var(\Sigma)\subseteq D$ and let $\Delta:=\Cl(\Sigma)$. Then the graph $G(\Delta)$ has the following properties.
\begin{itemize}
\item[(i)] Each vertex has a black, blue, and red self-loop.
\item[(ii)] The black, blue, and red subgraphs of $G(\Delta)$   are transitively closed. 
\item[(iii)] The black subgraph of $G(\Delta)$ is a subgraph of the blue subgraph of $G(\Delta)$. 
\item[(iv)] The subgraphs induced by the strongly connected components of $G(\Delta)$   are undirected. Each such component contains a black, blue, and red undirected subgraph. In the subgraph of each color, the vertices of the component can be partitioned into a collection of disjoint cliques.
All the vertices of the component belong to a single blue clique.
\item[(v)] If $\dep(\bar{x},y)\in\Delta$ and the vertices $\bar{x}$ have a common ancestor $z$ in the red subgraph of $G(\Delta)$, then there is a red edge from $z$ to $y$.
\end{itemize}
\end{lemma} 

We again assign a unique scc-number to each strongly connected component of $G(\Delta)$ such that the scc-number of a descendant component is always greater than the number of its ancestor component.

\begin{definition}[\cite{hirvonen24}]\label{addcan_def}
Let $\Sigma$ and $G(\Delta)$ be as in Lemma \ref{graph_prop_lemma} and assign an scc-numbering to  $G(\Delta)$. Define a team $X$ over the variables $D$ as follows.
\begin{itemize}
\item[(0)]  Let $\Delta'$ be a set of all nonunary functional dependencies $\dep(\bar{x},\bar{y})\not\in\Delta$ where $\Var(\bar{x}\bar{y})\subseteq D$. 
\begin{itemize}
    \item[(a)] For each $C\subseteq D$ with $|C|\geq 2$, check if there is $\dep(\bar{x},\bar{y})\in\Delta'$ such that $\Var(\bar{x})=C$. If yes, add a tuple with zeroes in exactly those positions $z$ that are functionally determined by $\bar{x}$, i.e., those $z$ for which $\dep(\bar{x},z)\in\Delta$. Leave all the other positions in the tuple empty.
\end{itemize} 
\item[(i-iii)] As in Definition \ref{uindfd_def}.
\item[(iv)] After all strongly connected components have been processed, we start filling the empty positions. Process again each strongly connected component in turn, starting with the one with the smallest scc-number. For each strongly connected component, count the number of maximal black cliques. If there are $n$ such cliques, number them from 0 to $n-1$. Then handle each maximal black clique $k$, $0\leq k\leq n-1$ in turn.
\begin{itemize}
\item[(a)] For each column in clique $k$, count the number of empty positions. If the column has $d>0$ empty positions, fill them with numbers $1,\dots d-1,d+k$ without repetitions. (Note that each column in $k$ has the same number of empty positions.) 
\end{itemize}
\item[(v)] If there are variables $x$ for which $\dep(x)\in\Delta$, they are all in the last component. The corresponding columns contain only zeroes and have no empty positions. As before, count the number of maximal black cliques. If there are $n$ such cliques, number them from 0 to $n-1$. Then handle each maximal black clique $k$, $0\leq k\leq n-1$ in turn.
\begin{itemize}
\item[(a)] For each column in clique $k$, change all the zeroes into $k$'s. 
\end{itemize}
\end{itemize}
\end{definition}
The point of step (0) is to ensure that $X$ does not satisfy any nonunary FDs that are not in $\Delta$. The other steps ensure that $X$ satisfies the nonunary FDs in $\Delta$, and that $X$ satisfies exactly those unary FDs or constant atoms that are in $\Delta$. The counting and filling in the empty positions is done to make sure that the number appearances of each element in a column is the same for exactly those variables that are in the same black clique, and, similarly, the multiset of the elements in each column is the same for exactly those variables that are in the same blue clique.

\section{Proof of Lemma \ref{IAUMIUMDE_lemma}}\label{proof}
We recall the proof of the lemma that was used in the proof of Theorem \ref{IAUMIUMDE_thm}.\\~\

\noindent\textbf{Lemma \ref{IAUMIUMDE_lemma}. (\cite{hirvonen24conf})}
        \emph{Let $\Sigma\cup\{\tau\}$ be a set of 
        UMIs, UMDEs, and IAs 
         such that $\Var(\Sigma\cup\{\tau\})\subseteq D$ and $\tau\notin\Cl(\Sigma)$. Then there exists 
         a $\mathbb{R}_{\geq0}$-team $\mathbb{X}:X\to\mathbb{R}_{\geq0}$ over $D$ such that $\mathbb{X}(s)=1/|X|$ for all $s\in X$, and $\mathbb{X}\models\Sigma$, but $\mathbb{X}\not\models\tau$.}
\begin{proof}
    \textbf{(UMI):} Suppose that $\tau:=x\approx y$. We construct a team $\X$ such that $\X\models \Sigma$ and $\X\not\models x\approx y$. Let $z_1,\dots,z_n$ be a list of those variables $z_i\in D$ for which $\Sigma\vdash x\approx^* z_i$, and let $u_1,\dots,u_m$ be a list of those variables $u_i\in D$ for which $\Sigma\not\vdash x\approx^* u_i$. The lists are clearly disjoint and $D=\{z_1,\dots,z_n\}\cup\{u_1,\dots,u_m\}$.
Define a team $X=\{s\}$, where 
\[
s(v)=
\begin{cases}
0, &\text{ if } v\in\{z_1,\dots,z_n\}\\
1, &\text{ if } v\in\{u_1,\dots,u_m\},
\end{cases}
\]
and let $\X\colon X\to \mathbb{R}_{\geq0}$ be such that $\X(s)=1$. Since $\Sigma\vdash x\approx x$ and $\Sigma\not\vdash x\approx y$, we have $x\in\{z_1,\dots,z_n\}$ and $y\in\{u_1,\dots,u_m\}$. Hence, by the construction, $\X\not\models x\approx y$. Suppose that $\Sigma\vdash v\approx v'$. Now, because of the transitivity axiom UMI3, either $v,v'\in \{z_1,\dots,z_n\}$ or $v,v'\in \{u_1,\dots,u_m\}$. This means that $\X\models v\approx v'$. It is easy to see that all UMDEs and IAs are satisfied by $\X$, so $\X\models\Sigma$.

\textbf{(UMDE):} Suppose that $\tau:= x\approx^* y$. We construct a team $\X$ such that $\X\models \Sigma$ and $\X\not\models x\approx^* y$. 
First, note that $\Sigma\vdash x\perp x$ and $\Sigma\vdash y\perp y$ imply $\Sigma\vdash x\approx^* y$, so either $\Sigma\not\vdash x\perp x$ or $\Sigma\not\vdash y\perp y$. Without loss of generality, assume that $\Sigma\not\vdash x\perp x$.

Let $z_1,\dots,z_n$ be a list of those variables $z_i\in D$ for which $\Sigma\vdash x\approx^* z_i$, and let $u_1,\dots,u_m$ be a list of those variables $u_i\in D$ for which $\Sigma\not\vdash x\approx^* u_i$. These lists are clearly disjoint and $D=\{z_1,\dots,z_n\}\cup\{u_1,\dots,u_m\}$.
Let the team $X$ consist of all the tuples from the set $Z_1\times\dots\times Z_n\times U_1\times\dots\times U_m$, where $Z_i=\{0,1\}$ and $U_j=\{0\}$ for $i=1,\dots,n$ and $j=1,\dots,m$, and define $\X$ as the uniform distribution over $X$. It is easy to see that $\X\not\models x\approx^* y$. Suppose that $\Sigma\vdash v\approx v'$. Now, since $v\approx v'$ implies $v\approx^* v'$, because of the transitivity axiom UMDE3, either $v,v'\in \{z_1,\dots,z_n\}$ or $v,v'\in \{u_1,\dots,u_m\}$. This means that $\X\models v\approx v'$. The case $\Sigma\vdash v\approx^* v'$ is analogous. 

Suppose that $\Sigma\vdash v\perp v$. Suppose for a contradiction that $\X\not\models v\perp v$. Then by the construction of $\X$, $v\in\{z_1,\dots,z_n\}$. This means that $\Sigma\vdash x\approx^* v$, and thus, by applying the axiom IA \& UMDE 1, we obtain $\Sigma\vdash x\perp x$. This contradicts the assumption that $\Sigma\not\vdash x\perp x$.
Suppose then that $\Sigma\vdash \bar{w}\pia \bar{w}'$. If $\Var(\bar{w})\cap\Var(\bar{w}')\neq\emptyset$, then by PIA3, $\Sigma\vdash v\perp v$ for all $v\in\Var(\bar{w})\cap\Var(\bar{w}')$. By the previous
case, we know that then $\mathbb{X}\models v\perp v$. Thus, by the axiom PIA5, we
may assume that $\Var(\bar{w})\cap\Var(\bar{w}')=\emptyset$.
It is easy to see by the construction that then $\X\models \bar{w}\pia \bar{w}'$.

\textbf{(IA):} Suppose first that $\tau:=x\perp x$. Then the construction from the UMDE case is such that $\X\models\Sigma$ and $\X\not\models x\perp x$. Suppose then that $\tau:= \bar{x}\pia\bar{y}$.  First, assume that $\Var(\bar{x})\cap\Var(\bar{y})\neq\emptyset$. Let $v\in\Var(\bar{x})\cap\Var(\bar{y})$. If $\Sigma\not\models v\perp v$,then $\Sigma\not\vdash v\perp v$ already
follows from the previous case, and thus by the axiom PIA3, we
also have $\Sigma\not\models \bar{x}\perp \bar{y}$. If $\Sigma\vdash v\perp v$, we can use the axiom PIA5 to infer
$\Sigma\not\vdash \bar{x}'\perp \bar{y}'$ where $\bar{x}'$ and $\bar{y}'$ are obtained from $\bar{x}$ and $\bar{y}$ by removing $v$. Hence, we
may assume that $\Var(\bar{x})\cap\Var(\bar{y})=\emptyset$.

We now construct a team $\X$ such that $\X\models\Sigma$ and $\X\not\models \bar{x}\pia\bar{y}$. 
We may assume that the atom $\bar{x}\pia \bar{y}$ is minimal in the sense that $\Sigma\vdash \bar{x}'\pia\bar{y}'$ for all $\bar{x}'$, $\bar{y}'$ such that $\Var(\bar{x}')\subseteq \Var(\bar{x})$, $\Var(\bar{y}')\subseteq \Var(\bar{y})$, and $\Var(\bar{x}'\bar{y}')\neq\Var(\bar{x}\bar{y})$. If not, we can remove variables from $\bar{x}$ and $\bar{y}$ until this holds. By the decomposition axiom IA3, it suffices to show the claim for the minimal atom. Note that due to the trivial independence axiom IA1 both $\bar{x}$ and $\bar{y}$ are at least of length one.

Let $\bar{x}=x_1\dots x_n$ and $\bar{y}=y_1\dots y_m$. Note that by the minimality of $\bar{x}\pia\bar{y}$, we have $\Sigma\not\vdash x_i\perp x_i$ and $\Sigma\not\vdash y_j\perp y_i$ for all $i=1,\dots,n$ and $j=1,\dots,m$. Let $\{u_1,\dots,u_k\}\subseteq D$ be the set of variables for which $\Sigma\vdash u_i\perp u_i$. Let $\{z_1,\dots,z_l\}=D\backslash(\{x_1,\dots,x_n\}\cup\{y_1,\dots,y_m\}\cup\{u_1,\dots,u_k\})$. Define a team $X_0$ over $D\backslash\{x_1\}$ such that it consists of all the tuples from the set $X_2\times\dots\times X_n\times Y_1\times\dots Y_m\times Z_1\times\dots\times Z_l\times U_1\times\dots\times U_k$, where $X_2=\dots = X_n = Y_1=\dots Y_m= Z_1=\dots = Z_l=\{0,1\}$ and $U_1=\dots = U_k=\{0\}$. Let then $X=\{s\cup\{(x_1,a)\}\mid s\in X_0\}$, where $a=\sum_{i=2}^{n} s(x_i)+\sum_{i=1}^{m} s(y_j) \ (\text{mod } 2)$.
Define then $\X$ as the uniform distribution over $X$.
Now $\X\not\models\bar{x}\pia\bar{y}$. Let $s, s'\in X$ be such that $s(x_1)=1$, $s(x_i)=0$, and $s'(y_j)=0$ for all $2\leq i\leq n$ and $1\leq j\leq m$. Then there is no $s''\in X$ such that $s''(\bar{x})=s(\bar{x})$ and $s''(\bar{y})=s'(\bar{y})$.

Suppose that $\Sigma\vdash v\approx v'$. Suppose for a contradiction that $\X\not\models v\approx v'$. Then one of $v$ and $v'$ must be in $\{u_1,\dots,u_k\}$ and one in $D\backslash\{u_1,\dots,u_k\}$.
Assume that $v'\in\{u_1,\dots,u_k\}$. This means that $\Sigma\vdash v'\perp v'$. Since $\Sigma\vdash v\approx v'$, by applying UMI \& UMDE and IA \& UMDE 1, we obtain $\Sigma\vdash v\perp v$. But then $v\in\{u_1,\dots,u_k\}$, which is a contradiction. The case $\Sigma\vdash v\approx^* v'$. is analogous.
Note also that if $\Sigma\vdash v\perp v$, then $v\in\{u_1,\dots,u_k\}$, and thus $\X\models v\perp v$.

Suppose then that $\Sigma\vdash \bar{w}\pia \bar{w}'$. We may assume that $u_i\not\in\Var(\bar{w}\bar{w}')$ for all $1\leq i\leq l$.
Assume first that $\Var(\bar{w}\bar{w}')\cap\Var(\bar{x}\bar{y})=\emptyset$. Then $\Var(\bar{w}\bar{w}')\subseteq\Var(\bar{z})$, where $\bar{z}=z_1\dots z_l$. It is clear from the definition of $\X$ that $\X\models \bar{w}\pia \bar{w}'$.

Assume then that $\Var(\bar{w}\bar{w}')\cap\Var(\bar{x}\bar{y})\neq\emptyset$, but $\Var(\bar{x}\bar{y})\not\subseteq\Var(\bar{w}\bar{w}')$. Then  $\X\models \bar{w}\pia \bar{w}'$ because $|\X_{\bar{w}=\bar{a}}|=(1/2)^{|\bar{w}|}$ for all $a\in\{0,1\}^{|\bar{w}|}$ and $|\X_{\bar{w}'=\bar{a}}|=(1/2)^{|\bar{w}'|}$ for all $a\in\{0,1\}^{|\bar{w}'|}$.

Assume finally that $\Var(\bar{x}\bar{y})\subseteq\Var(\bar{w}\bar{w}')$. We show that this case is not possible. We may assume that $\bar{w}=\bar{x}'\bar{y}'\bar{z}'$ and $\bar{w}'=\bar{x}''\bar{y}''\bar{z}''$, where $\Var(\bar{x})=\Var(\bar{x}'\bar{x}'')$, $\Var(\bar{y})=\Var(\bar{y}'\bar{y}'')$, and $\Var(\bar{z}'\bar{z}'')\subseteq\Var(\bar{z})$. By the axiom IA3, we have $\Sigma\vdash \bar{x}'\bar{y}'\pia \bar{x}''\bar{y}''$. Note that by the minimality of $\bar{x}\pia\bar{y}$, we have $\Sigma\vdash \bar{x}'\pia \bar{y}'$. Using the exchange axiom IA4 to $\bar{x}'\pia\bar{y}'$ and $\bar{x}'\bar{y}'\pia \bar{x}''\bar{y}''$, we obtain $\Sigma\vdash \bar{x}'\pia \bar{y}'\bar{x}''\bar{y}''$. So now, $\Sigma\vdash \bar{x}'\pia \bar{y}\bar{x}''$, and by the symmetry axiom IA2, $\Sigma\vdash  \bar{y}\bar{x}''\pia\bar{x}'$. Again, by the minimality of $\bar{x}\pia\bar{y}$ and the symmetry axiom IA2, we have $\Sigma\vdash \bar{y}\pia \bar{x}''$. Then using the exchange axiom IA4 again, this time to $\bar{y}\pia\bar{x}''$ and $\bar{y}\bar{x}''\pia\bar{x}'$, we obtain $\Sigma\vdash \bar{y}\pia \bar{x}''\bar{x}'$. By the symmetry axiom IA2, we have $\Sigma\vdash\bar{x}\pia\bar{y}$, which is a contradiction.
\end{proof}

\subsubsection*{Acknowledgements}
The author thanks all the anonymous referees for their valuable comments regarding 
the previous conference version of the article \cite{hirvonen26}. The author was supported by the Magnus Ehrnrooth foundation.

\bibliographystyle{cas-model2-names}
\bibliography{biblio}




\end{document}